\newcommand{\inbody}[1]{}

\documentclass[11pt]{article}

\usepackage{fullpage}
\usepackage{amsmath}
\usepackage{amssymb}
\usepackage{amsthm}
\usepackage{graphicx}
\usepackage{psfrag}
\usepackage{verbatim}
\usepackage{url}
\usepackage[usenames,dvipsnames]{xcolor}
\usepackage{hyperref}

\def\E{\mathop{\mathbb{E}}}
\def\R{\mathbb{R}}

\def\eps{\epsilon}
\def\del{\delta}
\newcommand{\D}{{\mathcal D}}
\newcommand{\Z}{{\mathcal Z}}
\newcommand{\A}{{\mathcal A}}
\newcommand{\B}{{\mathcal B}}
\newcommand{\hide}[1]{}

\def\F{\mathcal {F}}

\def\1{\mathbf{1}}

\def\F{\mathcal {F}}

\newcommand{\VSTAT}{{\mbox{VSTAT}}}

\newcommand{\etal}{{\em et al.\ }}

\newcommand{\MVSTAT}{{\mbox{MVSTAT}}}
\newcommand{\MSAMPLE}{{\mbox{1-MSTAT}}}
\newcommand{\cS}{\mathcal{S}}
\newcommand{\dc}{\kappa_2}
\newcommand{\SDN}{{\mathrm{SDN}}}

\newcommand{\SAMPLE}{{\mbox{1-STAT}}}
\newcommand{\ra}{\rangle}
\newcommand{\la}{\langle}
\newcommand{\sgn}{\mathsf{sign}}

\newcommand{\on}{\{\pm 1\}}
\DeclareMathOperator*{\argmin}{arg\,min}
\newcommand{\equ}[1]{
\begin{equation}
#1
\end{equation}}

\newcommand{\alequ}[1]{\begin{align} #1 \end{align}}
\newcommand{\alequn}[1]{\begin{align*} #1 \end{align*}}
\newcommand{\cond}{\ |\ }

\newtheorem{thm}{Theorem}[section]
\newtheorem{lem}[thm]{Lemma}
\newtheorem{rem}[thm]{Remark}
\newtheorem{cor}[thm]{Corollary}
\newtheorem{defn}[thm]{Definition}
\newtheorem{definition}[thm]{Definition}
\newtheorem{prop}[thm]{Proposition}
\newtheorem{conj}[thm]{Conjecture}

\newcommand{\paramini}[1]{
\smallskip
\noindent {\bf #1}
}

\begin{document}

\title{On the Complexity of Random Satisfiability Problems\\ with Planted Solutions\footnote{An extended abstract of this paper appeared at STOC, 2015 \cite{feldman2015complexity}.}}
\author{Vitaly Feldman\thanks{Google Research. Work done while at IBM Research and visiting the Simons Institute, UC Berkeley}
\and Will Perkins\thanks{University of Birmingham. Part of the work done while at Georgia Tech. Supported in part by an NSF postdoctoral fellowship.}
\and Santosh Vempala\thanks{Georgia Tech. Supported in part by NSF award CCF-1217793.}
}
\date{}
\maketitle

\begin{abstract}
The problem of identifying a planted assignment given a random $k$-SAT formula consistent with the assignment exhibits a large algorithmic gap: while the planted solution becomes unique and can be identified given a formula with $O(n\log n)$ clauses, there are distributions over clauses for which the best known
efficient algorithms require $n^{k/2}$ clauses. We propose and study a unified model for planted $k$-SAT, which captures well-known special cases. An instance is described by a planted assignment $\sigma$ and a distribution on clauses with $k$ literals. We define its {\em distribution complexity} as the largest $r$ for which the distribution is not $r$-wise independent ($1 \le r \le k$ for any distribution with a planted assignment).

Our main result is an unconditional lower bound, tight up to logarithmic factors, for {\em statistical} (query) algorithms \cite{kearns1998efficient,FeldmanGRVX:12}, matching  known upper bounds, which, as we show, can be implemented using a statistical algorithm.
Since known approaches for problems over distributions have statistical analogues (spectral, MCMC, gradient-based, convex optimization etc.), this lower bound provides a rigorous explanation of the observed algorithmic gap. The proof introduces a new general technique for the analysis of statistical query algorithms. It also points to a geometric {\em paring} phenomenon in the space of all planted assignments.

We describe consequences of our lower bounds to Feige's refutation
hypothesis \cite{feige2002relations}
and to lower bounds on general
convex programs that solve planted $k$-SAT. Our bounds also extend to other planted $k$-CSP models,
and, in particular, provide concrete evidence for the security of Goldreich's one-way function and the associated pseudorandom generator when used with a sufficiently hard predicate \cite{goldreich2000candidate}.

\end{abstract}

\thispagestyle{empty}
\newpage

\tableofcontents

\setcounter{page}{1}

\newpage


\section{Introduction}

Boolean satisfiability and constraint satisfaction problems are central to complexity theory; they are canonical NP-complete problems and their approximate versions are also hard. Are they easier on average for natural distributions? An instance of random satisfiability is generated by fixing a distribution over clauses, then drawing i.i.d.~clauses from this distribution. The average-case complexity of satisfiability problems is also motivated by its applications to models of disorder in physical systems, and to cryptography, which requires problems that are hard on average.

Here we study {\em planted satisfiability}, in which an assignment is fixed in advance, and clauses are selected from a distribution defined by the planted assignment.  Planted satisfiability and, more generally, random models with planted solutions appear widely in several different forms such as network clustering with planted partitions (the stochastic block model and its variants), random $k$-SAT with a planted assignment, and a proposed one-way function from cryptography \cite{goldreich2000candidate}.

It was noted in \cite{barthel2002hiding} that drawing satisfied $k$-SAT clauses uniformly at random from all those satisfied by an assignment  $\sigma \in \{ \pm 1\}^n$ often does not result in a difficult instance of satisfiability even if the number of observed clauses is relatively small. However, by changing the proportions of clauses depending on the number of satisfied literals under $\sigma$, one can create more challenging distributions over instances. Such ``quiet plantings" have been further studied in \cite{jia2005generating,achlioptas2005hiding,krzakala2009hiding,krzakala2012reweighted}.
Algorithms for planted $3$-SAT with various relative proportions were given by Flaxman \cite{flaxman2003spectral} and Coja-Oghlan \etal \cite{coja2010efficient}, the first of which works for $\Theta( n\log n)$ clauses but excludes distributions close to $3$-XOR-SAT, and the second of which works for all planted $3$-SAT distributions but requires $\Theta(n^{3/2} \ln^{10} n)$ clauses (note that a satisfiable $k$-XOR-SAT formula can be viewed as a satisfiable $k$-SAT formula with the same literals since XOR implies OR).
As $k$ increases, the problem exhibits a larger algorithmic gap:  the number of clauses required by known algorithms to efficiently identify a planted assignment is $\Omega(n^{k/2})$ while the number at which the planted assignment is the unique satisfying assignment is $O(n \log n)$.

We give a simple model for producing instances of planted $k$-SAT that generalizes and unifies past work on specific distributions for planted satisfiability. In this model, each clause $C$, a $k$-tuple of the $2n$ literals (variables and their negations), is included in the random formula with probability proportional to $Q(y)$ where $y \in \{\pm 1 \}^k$ is the value of the literals in $C$ on the planted assignment $\sigma$.
Here $Q$ can be an arbitrary probability distribution over $\{ \pm 1\}^k$. By choosing $Q$ supported only on $k$-bit strings with at least one true value, we can ensure that only satisfiable $k$-SAT formulas will be produced, but the model is more general and allows ``noisy" versions of satisfiability. We refer to an instance obtained by taking $Q$ to be uniform over $k$-bit strings with an even number of 1's as $k$-XOR-SAT (since each clause also satisfies an XOR constraint).

We identify the parameter of $Q$ that determines (up to lower order terms) the number of clauses that existing efficient algorithms require. It is the largest $r$ such that the distribution $Q$ is $(r-1)$-wise independent but not $r$-wise. Equivalently, it is the size of the smallest non-empty subset of $k$ indices for which the discrete Fourier coefficient of $Q$ is nonzero. This is always an integer between $1$ and $k$ for any distribution besides the uniform distribution on all clauses. Known algorithms  use $\tilde{O}(n^{r/2})$ clauses in general to identify the planted solution (with the exception of special cases which can be solved using Gaussian elimination and other algebraic techniques; see more details below). In \cite{FeldmanPV14} we gave an algorithm based on a subsampled power iteration that uses $\tilde{O}(n^{r/2})$ clauses to identify the planted assignment for any $Q$.

Our general formulation of the planted $k$-SAT problem and the notion of distribution complexity reveal a connection between planted $k$-SAT and the problem of inverting a PRG based on Goldreich's candidate one-way function \cite{goldreich2000candidate}, for which the link between $r$-wise independence and algorithmic tractability was known before~\cite{mossel2006varepsilon,austrin2009approximation,bogdanov2009security,applebaum2012dichotomy}. In this problem for a fixed predicate $P:\on^k \rightarrow \{-1,1\}$, we are given access to samples from a distribution $P_\sigma$, for a planted assignment $\sigma \in \on^n$. A random sample from this distribution is a randomly and uniformly chosen ordered $k$-tuple of variables (without repetition) $x_{i_1},\ldots,x_{i_k}$ together with the value $P(\sigma_{i_1},\ldots,\sigma_{i_k})$. As in the problem above, the goal is to recover $\sigma$ given $m$ random and independent samples from $P_\sigma$ or at least to be able to distinguish any planted distribution from one in which the value is a uniform random coin flip (in place of $P(\sigma_{i_1},\ldots,\sigma_{i_k})$). The number of evaluations of $P$ for which the problem remains hard determines the {\em stretch} of the pseudo-random generator (PRG). We note that despite the similarities between these two types of planted problems, we are not aware of any reductions between them (in Section \ref{sec:gen-planted} we show some relationships between these models and an even more general planted CSP model of Abbe and Montanari \cite{AbbeM15}).

Bogdanov and Qiao \cite{bogdanov2009security} show that an SDP-based algorithm of Charikar and Wirth \cite{CharikarWirth:04} can be used to find the input (which is the planted assignment) for any predicate that is {\em not} pairwise-independent using $m=O(n)$ such evaluations. The same approach can be used to recover the input for any $(r-1)$-wise (but not $r$-wise) independent predicate using $O(n^{r/2})$ evaluations~\cite{applebaum2016cryptographic}.

 Another important family of algorithms for recovering the planted assignment in Goldreich's PRG is algebraic, based on Gaussian elimination and its generalizations~\cite{mossel2006varepsilon,ApplebaumLovett15}. These attack algorithms are not captured by the framework of statistical algorithms we work with in this paper.  While algebraic approaches also apply to planted satisfiability problems, almost all planting functions $Q$ (in a measure-theoretic sense) are resilient against such algorithms, and any planted satisfiability problem can be made resistant by adding an $\eps$-fraction of uniformly random constraints.

The assumption that recovering the planted assignment in this problem is hard for some predicate has been used extensively in complexity theory and cryptography \cite{Alekhnovich11a,goldreich2000candidate, ishai2008cryptography, applebaum2010public, applebaum2012pseudorandom}, and the hardness of a decision version of this planted $k$-CSP is stated as the DCSP hypothesis in \cite{barak2013optimality}.  Applebaum~\cite{applebaum2012pseudorandom} reduced the search problem (finding the planted assignment) to the decision problem (distinguishing the output from uniformly random).  Our lower bounds below will be for this second, a priori easier, task.

Nearly optimal integrality gaps for LP and SDP hierarchies were recently given for this problem \cite{o2013goldreich} (and references therein) for $\Omega(n^{r/2-\eps})$  evaluations of a predicate that is $(r-1)$-wise but not $r$-wise independent.  Goldreich's PRG is shown to be an $\eps$-biased generator in \cite{mossel2006varepsilon, applebaum2012dichotomy}, and lower bounds against DPLL-style algorithms are given in \cite{cook2009goldreich}.  Applebaum and Lovett~\cite{ApplebaumLovett15} give lower bounds against algebraic attacks in a framework based on polynomial calculus.  

For a survey of these developments, see~\cite{applebaum2016cryptographic}.

\subsection{Summary of results}
\label{sec:our-results}
For the planted $k$-SAT problems and the planted $k$-CSPs arising from Goldreich's construction we address the following question:
{\em How many random constraints are needed to efficiently recover the planted assignment?} 

For these problems we prove unconditional lower bounds for a broad class of algorithms. Statistical (query) algorithms, defined by Kearns in the context of PAC learning \cite{kearns1998efficient} and by Feldman \etal \cite{FeldmanGRVX:12} for general problems on distributions, are algorithms that can be implemented without explicit access to random clauses, only being able to estimate expectations of functions of a random constraint to a desired accuracy. Many of the algorithmic approaches used in machine learning theory and practice have been shown to be implementable using statistical queries (e.g.~\cite{blum1998polynomial,DunaganV08,BlumDMN:05,ChuKLYBNO:06,BalcanF15}; see \cite{Feldman16:easq} for a brief overview) including most standard approaches to convex optimization \cite{FeldmanGV:15}. Other common techniques such as Expectation Maximization (EM)~\cite{DempsterLR77}, MCMC optimization \cite{TannerW87,GelfandSmith90}, (generalized) method of moments \cite{Hansen:82},  and simulated annealing~\cite{KirkpatrickGV83, Cerny1985Thermodynamical} are also known to fit into this framework. The only known problem for which a superpolynomial separation between the complexity of statistical algorithms and the usual computational complexity is known is solving linear equations over a finite field (which can be done via Gaussian elimination).

The simplest form of algorithms that we refer to as statistical are algorithm that can be implemented using evaluations of Boolean functions on a random sample. Formally, for a distribution $D$ over some domain (in our case all $k$-clauses) $\SAMPLE$ oracle is the oracle that given any function $h: X \rightarrow \{0,1\}$ takes a random sample $x$ from $D$ and returns $h(x)$. While lower bounds for this oracle are easiest to state and interpret, the strongest form of our lower bounds is for algorithms that use $\VSTAT$ oracle defined in \cite{FeldmanGRVX:12}. $\VSTAT(t)$ oracle captures the information about the expectation of a given function that is obtained by estimating it on $t$ independent samples.
\begin{definition}
  Let $D$ be the input distribution over the domain $X$. For an integer parameter $t > 0$,  for any query  function $h: X \rightarrow [0,1]$, $\VSTAT(t)$ returns a value $v \in \left[p- \tau, p + \tau\right]$ where
$p = \E_D[h(x)]$ and $\tau = \max \left\{\frac{1}{t}, \sqrt{\frac{p(1-p)}{t}}\right\}$.
\end{definition}
This oracle is based on the well-known statistical query oracle defined by Kearns \cite{kearns1998efficient} that
uses the same tolerance $\tau$ for all query functions. The $\VSTAT(t)$ oracle corresponds more tightly to access to $t$ samples and allows us to prove upper and lower bounds that closely correspond to known algorithmic bounds.

We show that the distribution complexity parameter $r$ characterizes the number of constraints (up to lower order terms) that an efficient statistical algorithm needs to solve instances of either problem. For brevity we state the bound for the planted $k$-SAT problem but identical bounds apply to Goldreich's $k$-CSP.
Our lower bound shows that any polynomial time statistical algorithm needs $\tilde{\Omega}(n^r)$ constraints to even {\em distinguish} clauses generated from a distribution with a planted assignment from uniformly random constraints (the decision problem). In addition, exponential time is required if $\tilde{\Omega}(n^{r-\eps})$ clauses are used for any $\eps > 0$.

More formally, for a clause distribution $Q$ and an assignment $\sigma$ let $Q_\sigma$ denote the distribution over clauses proportional to $Q$ for the planted assignment $\sigma$ (see Section~\ref{sec:prelims} for a formal definition). Let $U_k$ denote the uniform distribution over $k$-clauses.
\begin{thm}\label{thm:lower-bound-intro}
Let $Q$ be a distribution over $k$-clauses of complexity $r$. Then any (randomized) statistical algorithm that, given access to a distribution $D$ that equals $U_k$ with probability $1/2$ and equals $Q_\sigma$ with probability $1/2$ for a randomly and uniformly chosen $\sigma \in \on^n$, decides correctly whether $D = Q_\sigma$ or $D=U_k$ with probability at least $2/3$ needs either:
\begin{enumerate}
\item $\Omega(q)$ calls to $\VSTAT(\frac{n^{r}}{(\log q)^{r}})$ for any $q \geq 1$, or,
\item $\Omega((\frac{n}{\log n})^r)$ calls to $\SAMPLE$.
\end{enumerate}
\end{thm}

It is easy to see that this lower bound is essentially tight for statistical algorithms using the $\VSTAT$ oracle (since noisy $r$-XOR-SAT can be solved using a polynomial (in $n^r$) number of queries to $\VSTAT(O(n^r))$ that can determine the probability of each clause). Surprisingly, this lower bound is quadratically larger than the upper bound of $\tilde{O}(n^{r/2})$ that can be achieved using samples themselves \cite{FeldmanPV14}. While unusual, this is consistent with a common situation where an implementation using a statistical oracle requires polynomially more samples (for example in the case of algorithms for learning halfspaces). Still this discrepancy is an interesting one to investigate in order to better understand the power of statistical algorithms and lower bounds against them. We show that there exist natural strengthenings of the $\VSTAT$ and $\SAMPLE$ oracles that bridge this gap. Specifically, we extend the oracles to functions with values in a larger discrete range $\{0,1,\ldots,L-1\}$ for $L \geq 2$: $\MSAMPLE(L)$ oracle is the oracle that given any function $h: X \rightarrow \{0,1,\ldots,L-1\}$ takes a random sample $x$ from $D$ and returns $h(x)$ and $\VSTAT$ is extended similarly to $\MVSTAT$ (we postpone the formal details and statements for this oracle to Section \ref{sec:define-stat-oracles}). This strengthening interpolates between the full access to samples which corresponds to $L = |X_k|$ and the standard statistical query oracles (corresponding to $L=2$) and hence is a natural one to investigate.

We prove nearly matching upper and lower bounds for the stronger oracle: $(a)$ there is an efficient statistical algorithm that uses $\tilde O(n^{r/2})$ calls to $\MSAMPLE(O(n^{\lceil r/2\rceil}))$ and identifies the planted assignment; $(b)$ there is no algorithm that can solve the problem described in Theorem \ref{thm:lower-bound-intro} using less than $\tilde O(n^{r/2})$ calls to $\MSAMPLE(n^{r/2})$. We state the upper bound more formally:
\begin{thm}\label{thm:algo1-intro}
Let $Q$ be a clause distribution of distribution complexity $r$. Then there exists an algorithm
that uses $O(n^{r/2} \log^2 n)$ calls to $\MSAMPLE(n^{\lceil r/2\rceil})$ and time linear in the number of oracle calls to identify the planted assignment with probability $1-o(1)$.
\end{thm}
We prove this bound by showing that the algorithm from \cite{FeldmanPV14} based on a subsampled power iteration can be implemented using statistical query oracles. The same upper bound holds for Goldreich's planted $k$-CSP.

 In addition to providing a matching lower bound, the algorithm gives an example of statistical query algorithm for performing power iteration to compute eigenvectors or singular vectors. Spectral algorithms are among the most commonly used for problems with planted solutions (including Flaxman's algorithm \cite{flaxman2003spectral} for planted satisfiability) and our lower bounds can be used to derive lower bounds against such algorithms. The alternative approach for solving planted constraint satisfaction problems with $O(n^{r/2})$ samples is to use an SDP solver as shown in \cite{bogdanov2009security} (with the ``birthday paradox" as shown in \cite{o2013goldreich}; see also~\cite{applebaum2016cryptographic}). This approach can also be implemented using statistical queries, although
 a direct implementation using a generic SDP solver such as the one we describe in Section \ref{sec:statalgpowersec} will require quadratically more samples and will not give a non-trivial statistical algorithm for the problem (since solving using $O(n^r)$ clauses is trivial).

We now briefly mention some of the corollaries and applications of our results.

\subsubsection{Evidence for Feige's hypothesis:}
A closely related problem is refuting the satisfiability of a random $k$-SAT formula (with no planting), a problem conjectured to be hard by Feige \cite{feige2002relations}. A refutation algorithm takes a $k$-SAT formula $\Phi$ as an input and returns either SAT or UNSAT. If $\Phi$ is satisfiable, the algorithm always returns SAT and for $\Phi$ drawn uniformly at random from all $k$-SAT formulas of $n$ variables and $m$ clauses the algorithm must return UNSAT with probability at least $2/3$. For this refutation problem, an instance becomes unsatisfiable w.h.p.~after $O(n)$ clauses, but algorithmic bounds are as high as those for finding a planted assignment under the noisy XOR distribution: $O(n^{k/2})$ clauses suffice\cite{friedman2005recognizing, coja2004techniques, han2009note, goerdt2003recognizing, feige2004easily, AllenOW15}.

To relate this problem to our lower bounds we define an equivalent distributional version of the problem. In this version the input formula is obtained by sampling $m$ i.i.d. clauses from some unknown distribution $D$ over clauses. The goal is to say UNSAT (with probability at least $2/3$) when clauses are sampled from the uniform distribution and to say SAT for every distribution supported on simultaneously satisfiable clauses.

In the distributional setting, an immediate consequence of Theorem \ref{thm:lower-bound-intro} is that Feige's hypothesis holds for the class of statistical query algorithms. The proof (see Theorem~\ref{cor:feige}) follows from the fact that our decision problem (distinguishing between a planted $k$-SAT instance and the uniform $k$-SAT instance) is a special case of the distributional refutation problem.

\subsubsection{Hard instances of $k$-SAT:}
Finding distributions of planted $k$-SAT instances that are algorithmically intractable has been a pursuit of researchers in both computer science and physics. The distribution complexity parameter defined here generalizes the notion of ``quiet plantings" studied in physics \cite{barthel2002hiding,jia2005generating,krzakala2009hiding,krzakala2012reweighted}
to an entire hierarchy of ``quietness".  In particular, there are easy to generate distributions of satisfiable $k$-SAT instances with distribution complexity as high as $k-1$ ($r=k$ can be achieved using XOR constraints but these instances are solvable by Gaussian elimination).  These instances can also serve as strong tests of industrial SAT solvers as well as the underlying hard instances in cryptographic applications. In recent work, Blocki \etal extended our lower bounds from the Boolean setting to $\Z_d$ and applied them to show the security of a class of humanly computable password protocols \cite{BBDV14}.

\subsubsection{Lower bounds for convex programs:} Our lower bounds imply limitations of using convex programs to recover
planted solutions. For example, any convex program whose objective is the sum of
objectives for individual constraints (as is the case for canonical
LPs/SDPs for CSPs) and distinguishes between a planted CSP instance and a uniformly generated one must have dimension at least $\tilde{\Omega}(n^{r/2})$. In
particular, this lower bound applies to lift-and-project hierarchies
where the number of solution space constraints increases (and so does the cost of
finding a violated constraint), but the dimension remains the same.
Moreover, since our bounds are for detecting planted solutions, they
imply large integrality gaps for convex relaxations of this dimension.
These bounds follow from statistical implementations of
algorithms for convex optimization given in \cite{FeldmanGV:15}. We emphasize that the lower bounds apply to convex relaxations themselves and make no assumptions on how the convex relaxations are solved (in particular, the solver does not need to be a statistical algorithm). An example of such lower bound is given below. Roughly speaking, the corollary says that any convex program whose objective value
is significantly higher for the uniform distribution over clauses, $U_k$, compared to a planted distribution $Q_\sigma$ must have a large dimension, independent of the number of constraints.
\begin{cor}
\label{cor:lower-convex-program-intro}
Let $Q$ be a distribution over $k$-clauses of complexity $r$. Assume that there exists a mapping that maps each $k$-clause $C\in X_k$ to a convex function $f_C:K \rightarrow[-1,1]$ over some bounded, convex and compact $N$-dimensional set $K$. Further assume that for some $\eps > 0$ and $\alpha \in \R$:
$$\Pr_{\sigma \in \on^n}\left[ \min_{x\in K}\left\{ \E_{C\sim Q_\sigma}[ f_{C}(x)]  \right\} \leq \alpha\right] \geq 1/2.$$
and
$$\min_{x\in K} \left\{\E_{C\sim U_k}[f_{C}(x)] \right\} > \alpha + \eps .$$
Then $N = \tilde{\Omega}\left(n^{r/2}\cdot \eps \right)$.
\end{cor}
We note that conditions on the value of the convex program that we make are weaker than the standard conditions that a convex relaxation must satisfy. Specifically, it is usually assumed that a convex relaxation does not increase the value of the objective (for example, the value for a satisfiable instance must be 0) and also that the minimum of the objective function for all ``bad" instances will be noticeably larger than that of the ``good" instances. In Section \ref{sec:statalgpowersec} we also prove lower bounds against convex programs in exponentially high dimension as long as the appropriate norms of points in the domain and gradients are not too large. We are not aware of this form of lower bounds against convex programs for planted satisfiability stated before. We also remark that our lower bounds are incomparable to lower bounds for programs given in \cite{o2013goldreich} since they analyze a specific SDP for which the mapping $\cal M$ maps to functions over an $O(n^k)$-dimensional set $K$ is defined using a high level of the Sherali-Adams or Lov\'asz-Schrijver hierarchies.
Further details are given in Section \ref{sec:statalgpowersec}.

\subsection{Overview of the technique}
Our proof of the lower bound builds on the notion of \textit{statistical dimension} given in \cite{FeldmanGRVX:12} which itself is based on ideas developed in a line of work on statistical query learning \cite{kearns1998efficient,BlumFJ+:94,Feldman:12jcss}.

Our primary technical contribution is a new, stronger notion of statistical dimension and its analysis for planted $k$-CSP problems. The statistical dimension in \cite{FeldmanGRVX:12} is based on upper-bounding average or maximum pairwise correlations between appropriately defined density functions. While these dimensions can be used for our problem (and, indeed, were a starting point for this work) they do not lead to the tight bounds we seek. Specifically, at best they give lower bounds for $\VSTAT(n^{r/2})$, whereas we will prove lower bounds for $\VSTAT(n^r)$ to match the current best upper bounds.

Our stronger notion directly examines a natural operator,  which, for a given function, evaluates how well the expectation of the function discriminates between different distributions. We show that a norm of this operator for large sets of input distributions gives a lower bound on the complexity of any statistical query algorithm for the problem.
Its analysis for our problem is fairly involved and a key element of the proof is the use of concentration of polynomials on $\on^n$ (derived from the hypercontractivity results of Bonami and Beckner \cite{bonami1970etude,beckner1975inequalities}).

We remark that the $k$-XOR-SAT problem is equivalent to PAC learning of (general) parity functions from random $k$-sparse examples. The latter is the classic problem addressed by Kearns' original lower bound \cite{kearns1998efficient}. While
superficially the planted setting is similar to learning of $k$-sparse parities from random uniform examples for which optimal statistical query lower bounds are well-known and easy to derive, the problems, techniques and the resulting bounds are qualitatively different. One significant difference is that the correlation between parity functions on the uniform distribution is 0, whereas in our setting the distributions are not uniform and pairwise correlations between them can be relatively large. Moreover, as mentioned earlier, the techniques based on pairwise correlations do not suffice for the strong lower bounds we give.

Our stronger technique gives further insight into the complexity of statistical algorithms\hide{. In fact, a norm of the operator we introduce here gives a characterization of the complexity of statistical algorithms \cite{Feldman-characterization}. The technique also} and has a natural interpretation in terms of the geometry of the space of all planted assignments with a metric defined (between pairs of assignments) to capture properties of statistical algorithms.  The fraction of solutions that are at distance greater than some threshold from a fixed assignment goes up sharply from  exponentially small to a polynomial fraction as the distance threshold increases. We call this a `paring' transition as a large number of distributions become amenable to being separated from the planted solution and discarded.

We conjecture that our lower bounds hold for {\em all} algorithms with the exception of those based on Gaussian elimination. Formalizing ``based on Gaussian elimination" requires substantial care. Indeed, in an earlier version of this work we excluded Gaussian elimination by only excluding density functions of low algebraic degree. (Here algebraic degree refers to the degree of the polynomial over $\Z_2^k$ required to represent the function. For example, the parity function equals to $x_1 + x_2 + \cdots + x_k$ and therefore has algebraic degree $1$). This resulted in a conjecture that was subsequently disproved by Applebaum and Lovett \cite{ApplebaumLovett15} using an algorithm that combines Gaussian elimination with fixing some of the variables. An alternative approach to excluding Gaussian elimination-based methods is to exploit their fragility to even low rates of random noise. Here random noise would correspond to mixing in of random and uniform constraints to the distribution. In other words for  $\alpha\in [0,1]$, $Q$ becomes $Q^\alpha = (1-\alpha)Q + \alpha U_k$. Observe that for all constant $\alpha<1$, the complexity of $Q^\alpha$ is the same as the complexity of $Q$.
\begin{conj}\label{conj:mainconj}
Let $Q$ be any distribution over $k$-clauses of complexity $r$ and $\alpha \in (0,1)$. Then any polynomial-time (randomized) algorithm that, given access to a distribution $D$ that equals either $U_k$ or $Q_\sigma^\alpha$ for some $\sigma \in \on^n$, decides correctly whether $D = Q_\sigma^\alpha$ or $D=U_k$ with probability at least $2/3$ needs at least $\tilde \Omega(n^{r/2})$ clauses.
\end{conj}
We conjecture that an analogous statement also holds for Goldreich's $k$-CSP. Note that in this case mixing in an $\alpha$-fraction of random and uniform constraints can be equivalently seen as flipping the given value of the predicate with probability $\alpha/2$ randomly and independently for each constraint.

\subsection{Other related work}
\label{sec:prior}
\paramini{Hypergraph Partitioning.} Another closely related model to planted satisfiability is random hypergraph partitioning, in which a partition of the vertex set is fixed, then $k$-uniform hyperedges added with probabilities that depend on their overlap with the partition.  To obtain a planted satisfiability model from a planted hypergraph, let the vertex set be the set of $2n$ literals, with the partition given by the planted assignment $\sigma$.  A $k$-clause is then a $k$-uniform hyperedge.  The two models are not exactly equivalent, as in planted satisfiability we have the extra information that pairs of literals corresponding to the same variable must receive different assignments; however, to the best of our knowledge  all the known algorithmic approaches to planted satisfiability work for planted hypergraph partitioning as well.   The Goldreich CSP model is closely related to a hypergraph version of the censored block model~\cite{AbbeM15,abbe2014decoding} in which random hyperedges are labeled with values that depend on how the edges overlap with a planted partition.

  The case $k=2$ of $k$-uniform hypergraph partitioning is called the stochastic block model.  The input is a random graph with different edge probabilities within and across an unknown partition of the vertices, and the algorithmic task is to recover partial or complete information about the partition given the resulting graph.
Work on this model includes Bopanna \cite{boppana1987eigenvalues}, McSherry's general-purpose spectral algorithm \cite{mcsherry2001spectral}, and Coja-Oghlan's algorithm that works for graphs of constant average degree \cite{coja2006spectral}.

Recently an intriguing threshold phenomenon was conjectured by Decelle, Krzakala, Moore, and Zdeborov{\'a} \cite{decelle2011asymptotic}: there is a sharp threshold separating efficient partial recovery of the partition from information-theoretic impossibility of recovery. This conjecture was proved in a series of works \cite{mossel2015reconstruction,massoulie2014community,mossel2013proof}.  In the same work Decelle et al.~conjecture that for a planted $q$-coloring, there is a gap of algorithmic intractability between the impossibility threshold and the efficient recovery threshold. Neither lower bounds nor an efficient algorithm at their conjectured threshold are currently known.
    In our work we consider planted bipartitions of $k$-uniform hypergraphs, and show that the behavior is dramatically different for $k \ge 3$.  Here, while the information theoretic threshold is still at a linear number of hyperedges, we give evidence that the efficient recovery threshold can be much larger, as high as  $\tilde \Theta (n^{k/2})$.  In fact, our lower bounds hold for the problem of distinguishing a random hypergraph with a planted partition from a uniformly random one and thus give computational lower bounds for checking hypergraph quasirandomness (see \cite{Trevisan:2008:Online} for more on this problem).  Throughout the paper we will use the terminology of planted satisfiability (assignments, constraints, clauses) but all results apply also to random hypergraph partitioning.

\paramini{Shattering and paring.}
Random satisfiability problems (without a planted solution) such as $k$-SAT and $k$-coloring random graphs exhibit a shattering phenomenon in the solution space for large enough $k$ \cite{krzakala2007gibbs, achlioptas2008algorithmic}: as the density of constraints increases, the set of all solutions evolves from a large connected cluster to a exponentially large set of well-separated clusters.  The shattering threshold empirically coincides with the threshold for algorithmic tractability (while this is the case for large $k$, for $k=3,4$ there is some evidence that the survey propagation algorithm may succeed beyond the shattering threshold~\cite{marino2016backtracking}).  Shattering has also been used to prove that certain algorithms fail at high enough densities \cite{gamarnik2014performance}.

Both the shattering and paring phenomena give an explanation for the failure of known algorithms on random instances.  Both capture properties of local algorithms, in the sense that in both cases, the performance of Gaussian elimination, an inherently global algorithm, is unaffected by the geometry of the solution space: both random $k$-XOR-SAT and random planted $k$-XOR-SAT are solvable at all densities despite exhibiting shattering and paring respectively.

The paring phenomenon differs from shattering in several significant ways.  As the paring transition is a geometric property of a carefully chosen metric, there is a direct and provable link between paring and algorithmic tractability, as opposed to the empirical coincidence of shattering and algorithmic failure.  In addition, while shattering is known to hold only for large enough $k$, the paring phenomenon holds for all $k$, and already gives strong lower bounds for $3$-uniform constraints.

One direction for future work would be to show that the paring phenomenon exhibits a sharp threshold; in other words, improve the analysis of the statistical dimension of planted satisfiability in Section \ref{sec:plantstat-sdn} to remove the logarithmic gap between the upper and lower bounds.  An application of such an improvement would be to apply the lower bound framework to the planted coloring conjecture from \cite{decelle2011asymptotic}; as the gap between impossibility and efficient recovery is only a constant factor there, the paring transition would need to be located more precisely.

\section{Definitions}
\label{sec:prelims}

\subsection{Planted satisfiability}
\label{sec:models}
We now define a general model for  planted satisfiability problems that unifies various previous ways to produce a random $k$-SAT formula where the relative probability that a clause is included in the formula depends on the number of satisfied literals in the clause \cite{flaxman2003spectral,jia2005generating,achlioptas2005hiding,krivelevich2006solving,krzakala2009hiding,coja2010efficient,krzakala2012reweighted}.

Fix an assignment $\sigma \in \{ \pm 1 \}^n$. We represent a $k$-clause  by an ordered $k$-tuple of literals from $x_1, \dots x_n, \overline x_1, \dots \overline x_n$ with no repetition of variables and let $X_k$ be the set of all such $k$-clauses.
For a $k$-clause $C=(l_1,\ldots,l_k)$ let  $\sigma(C) \in \on^k$ be the $k$-bit string of values assigned by $\sigma$ to literals in $C$, that is $\sigma(l_1),\dots,\sigma(l_k)$, where
$\sigma(l_i)$ is the value of literal $l_i$ in assignment $\sigma$ with $-1$ corresponding to TRUE and 1 to FALSE. In a planted model, we draw clauses with probabilities that depend on the value of $\sigma(C)$.

A planted distribution $Q_\sigma$ is defined by a distribution
$Q$ over $\on^k$, that is a function $Q:\on^k \to \R^+$  such that
\[ \sum_{y\in\on^k} Q(y) = 1. \]
To generate a random formula, $F(Q,\sigma,m)$ we draw $m$ i.i.d. $k$-clauses according to the probability distribution $Q_\sigma$, where
\[
Q_\sigma(C)= \frac{ Q(\sigma(C))}{\sum_{C^\prime \in X_k} Q(\sigma(C^\prime))}
.\]
By concentrating the support of Q only on satisfying assignments of an appropriate predicate we can generate satisfiable distributions for any predicate, including $k$-SAT, $k$-XOR-SAT, and $k$-NAE-SAT.  In most previously considered distributions $Q$ is a symmetric function, that is $Q_\sigma$ depends only on the number of satisfied literals in $C$.
For brevity in such cases we define $Q$ as a function from $\{0, \dots k \}$ to $\R^+$.
For example, the planted uniform $k$-SAT distribution fixes one assignment $\sigma \in \{ \pm 1 \}^n$ and draws $m$ clauses uniformly at random conditioned on the clauses being satisfied by $\sigma$.  In our model, this corresponds to setting $Q(0) =0$ and $Q(i) = 1 / (2^k -1)$ for $i \ge 1$.  Planted $k$-XOR-SAT, on the other hand, corresponds to setting $Q(i) =0$ for $i$ even, and $Q(i) =1/2^{k-1}$ for $i$ odd.

\paramini{Problems.}
The algorithmic problems studied in this paper can be stated as follows:
Given the function $Q$ and a sample of $m$ independent clauses drawn according to $Q_\sigma$, recover $\sigma$, or some $\tau$ correlated with $\sigma$. Note that since unsatisfiable clauses are allowed to have non-zero weight, for some distributions the problem is effectively satisfiability with random noise. Our lower bounds are for the potentially easier problem of distinguishing a randomly and uniformly chosen planted distribution from the uniform distribution over $k$-clauses. Namely, let $\D_Q$ denote the set of all distributions $Q_\sigma$, where $\sigma \in \on^k$ and $U_k$ be the uniform distribution over $k$-clauses. Let $\B(\D_Q,U_k)$ denote the decision problem in which given samples from an unknown input distribution $D \in \D_Q \cup \{U_k\}$ the goal is to output $1$ if $D\in \D_Q$ and 0 if $D=U_k$.

In Goldreich's planted $k$-CSP problem for a predicate $P:\on^k \rightarrow \{-1,1\}$, we are given access to samples from a distribution $P_\sigma$, where $\sigma$ is a planted assignment in $\on^n$. A random sample from this distribution is a randomly and uniformly chosen ordered $k$-tuple of variables (without repetition) $x_{i_1},\ldots,x_{i_k}$ together with the value $P(\sigma_{i_1},\ldots,\sigma_{i_k})$. To allow for predicates with random noise and further generalize the model we also allow any real-valued $P:\on^k \rightarrow [-1,1]$. For such $P$, instead of the value $P(\sigma_{i_1},\ldots,\sigma_{i_k})$, a randomly and independently chosen value $b\in \{-1,1\}$ such that $\E[b] = P(\sigma_{i_1},\ldots,\sigma_{i_k})$ is output.

As in the problem above, the goal is to recover $\sigma$ given $m$ random and independent samples from $P_\sigma$ or at least to be able to distinguish any planted distribution from one in which the value is a uniform random coin flip (or, equivalently, the distribution obtained when the function $P \equiv 0$). Our goal is to understand the smallest number $m$ of $k$-clauses
 that suffice to find the planted assignment or at least to distinguish a planted distribution from a uniform one.

For a clause distribution $Q$, we define its {\em distribution complexity} $r(Q)$ as the smallest integer $r\ge 1$ for which there exists a set $S \subseteq [k]$ of size $r$ and
\begin{equation}
\label{eq:distcomplexity}
\hat{Q}(S) \doteq \frac{1}{2^k} \cdot \sum_{y \in \on^k}  \left[Q(y) \prod_{i \in S} y_i \right]  \neq 0.
\end{equation}
$\hat{Q}(S)$ is the Fourier coefficient of the function $Q$ on the set $S$ (see Section~\ref{sec:plantstat-sdn} for a formal definition).
For a symmetric function the value of $\hat{Q}(S)$ depends only on $|S|$ and therefore we refer to the value of the coefficient for sets of size $\ell$ by $\hat{Q}(\ell)$.

To see the difference between a hard and easy distribution $Q$, first consider  planted uniform $k$-SAT:  $Q(0)=0$, $Q(i) = 1/(2^k-1)$ for $i\ge 1$.  The distribution complexity of $Q$ is $r=1$.

Next, consider the noisy parity distribution (or noisy planted $k$-XOR-SAT) with $Q(l) = \del/2^{k-1}$ for $l$ even, and $Q(l) = (2- \del)/2^{k-1}$ for $l$ odd, for  some $\del \ne 1$.    In this case, we have
$\hat{Q}(l) =0$ for $1 \le l \le k-1$, and so the distribution complexity of $Q$ is $r=k$.
We will see that such parity-type distributions are in
fact the hardest for statistical algorithms to detect.
\subsection{Statistical algorithms}
\label{sec:define-stat-oracles}
We can define planted satisfiability as the problem of identifying an unknown distribution $D$ on a domain $X$ given $m$ independent samples from $D$.  For us, $X$ is the set of all possible $k$-clauses or $k$-hyperedges, and each partition or assignment $\sigma$ defines a unique distribution $D_\sigma$ over $X$.

Extending the work of Kearns \cite{kearns1998efficient} in learning theory, Feldman \etal \cite{FeldmanGRVX:12} defined statistical query algorithms for problems over distributions.
Roughly speaking, these are algorithms that do
not see samples from the distribution but instead have access to
estimates of the expectation of any bounded function of a sample from
the distribution. More formally, a statistical algorithm can access the input distribution via one of the following oracles.

\begin{definition}[$\MSAMPLE(L)$ oracle]
  Let $D$ be the input distribution over the domain $X$.  Given any function $h: X \rightarrow \{0,1,\ldots,L-1\}$,
  $\MSAMPLE(L)$ takes a random sample $x$ from $D$ and returns $h(x)$.
\end{definition}
This oracle is a generalization of the $\SAMPLE$ oracle from \cite{FeldmanGRVX:12} and was first defined by Ben-David and Dichterman in the context of PAC learning \cite{Ben-DavidD98}. It was also more recently studied in \cite{SteinhardtD15,Steinhardt:2015}. For the planted SAT problem this oracle allows an algorithm to evaluate a multi-valued function on a random clause. By repeating the query, the algorithm can estimate the expectation of the function as its average on independent samples. Being able to output one of multiple possible values gives the algorithm considerable flexibility, e.g., each value could correspond to whether a clause has a certain pattern on a subset of literals. With $L=n^{k}$, the algorithm can identify the random clause. We will therefore be interested in the trade-off between $L$ and the number of queries needed to solve the problem.

The next oracle is from \cite{FeldmanGRVX:12}.
\begin{definition}[$\VSTAT$ oracle]
  Let $D$ be the input distribution over the domain $X$. For an integer parameter $t > 0$,  for any query  function $h: X \rightarrow [0,1]$, $\VSTAT(t)$ returns a value $v \in \left[p- \tau, p + \tau\right]$ where
$p = \E_D[h(x)]$ and $\tau = \max \left\{\frac{1}{t}, \sqrt{\frac{p(1-p)}{t}}\right\}$.
\end{definition}
The definition of $\tau$ means that $\VSTAT(t)$ can return any value $v$ for which the distribution $B(t,v)$ (outcomes of $t$ independent Bernoulli variables with bias $v$) is close to $B(t,E[h])$ in total variation distance \cite{FeldmanGRVX:12}. In most cases $p>1/t$ and then $\tau$ also corresponds to returning the expectation of a function to within the standard deviation error of averaging the function over $t$ samples. However, it is important to note that within this constraint on the error, the oracle can return any value, possibly in an adversarial way.

In this paper, we also define the following generalization\footnote{For simplicity, this definition generalizes $\VSTAT$ only for Boolean query functions.} of the $\VSTAT$ oracle to multi-valued functions.
\begin{definition}[$\MVSTAT$ oracle]
 Let $D$ be the input distribution over the domain $X$, $t , L > 0$ be integers.
For any multi-valued function $h:X \rightarrow \{0,1,\ldots, L-1\}$ and any set $\cS$ of subsets of $ \{0,\ldots, L-1\}$, $\MVSTAT(L,t)$ returns a vector $v \in \R^L$ satisfying for every $Z \in \cS$
$$\left|\sum_{\ell \in Z} v_l - p_Z \right| \le \max \left\{\frac{1}{t}, \sqrt{\frac{p_Z(1-p_Z)}{t}}\right\} ,$$
where $p_Z = \Pr_D[h(x) \in Z]$.
The query cost of such a query is $|\cS|$.
\end{definition}
We note that $\VSTAT(t)$ is equivalent to $\MVSTAT(2,t)$ (the latter only allows Boolean queries but that is not an essential difference) and any query to $\MVSTAT(L,t)$ can be easily answered using $L$ queries to $\VSTAT(4\cdot Lt)$ (see Theorem~\ref{thm:simulate-mvstat} for a proof). The additional strength of this oracle comes from allowing the sets in $\cS$ to depend on the unknown distribution $D$ and, in particular, be fixed but unknown to the algorithm. This is useful for ensuring that potential functions of our discrete power iteration algorithm for planted SAT behave in the same way as if the algorithm were executed on true samples (see Section \ref{sec:statImplement}).
Another useful way to think of $L$-valued oracles in the context of vector-based algorithms is as a vector of $L$ Boolean functions which are non-zero on disjoint parts of the domain. This view also allows to extend $\MVSTAT$ to bounded-range (non-Boolean) functions.

An important property of every one of these oracles is that it can be easily simulated using $t$ samples (in the case of $\VSTAT$/$\MVSTAT$ the success probability is a positive constant but it can be amplified to $1-\delta$ using $O(t \log(1/\delta))$ samples).
\hide{

As we have mentioned above, these oracles (even for $L=2$) allow one to implement most known algorithmic techniques to problems over distributions (which include various planted solution problems), including gradient-based methods, convex optimization, MCMC methods, spectral methods etc.  The algorithm we describe below, for example, is a statistical algorithm for performing power iteration to compute eigenvectors or singular vectors.  Adapting it leads to statistical analogues of spectral-based algorithms for planted satisfiability, like Flaxman's algorithm \cite{flaxman2003spectral}.  One important issue is that such an implementation might require polynomially more samples than using samples themselves. This is, for example, the case for LP/SDP solving. For our problem this implies that implementing the existing algorithm using $O(n^{r/2})$ samples (based on the birthday ``paradox" and SDP \cite{bogdanov2009security}) will not give a non-trivial statistical algorithm for the problem since solving using $O(n^r)$ clauses is trivial. However, as we show here, there exists a statistical algorithm for the planted satisfiability problems using $t=\tilde{O}(n^{r/2})$ (the equivalent of the number of random clauses). Our algorithm requires access to $\MVSTAT(L,t)$ oracle or $t$ samples from $\MSAMPLE(L)$ oracle for $L = O(n^{\lceil r/2 \rceil})$.} The goal of our generalization of oracles to $L>2$ was to show that even nearly optimal sample complexity can be achieved by a statistical algorithm using an oracle for which a nearly matching lower bound applies.

\section{Results}
\label{sec:mainresults}

We state our upper and lower bounds for the planted satisfiability problem. Identical upper and lower bounds apply to Goldreich's planted $k$-CSPs with $r$ being the degree of lowest-degree non-zero Fourier coefficient of $P$. For brevity, we omit the repetitive definitions and statements in this section. In Section \ref{sec:gen-planted} we give the extension of our lower bounds to this problem and also make the connections between the two problems explicit.

\subsection{Lower bounds}
We begin with lower bounds for {\em any} statistical algorithm.
For a clause distribution $Q$ let $\B(\D_Q,U_k)$ denote the decision problem of distinguishing whether the input distribution is one of the planted distributions or is uniform.
\begin{thm}\label{thm:lower-bound}
For an assignment $\sigma \in \on^n$, let $Q_\sigma$ be a distribution over $k$-clauses of complexity $r$ and $\D_Q$ be this family of distributions. Assume that the input distribution $D$ is $U_k$ with probability $1/2$ and with the remaining probability $1/2$ it is $Q_\sigma$ for a uniform random $\sigma\in \on^n$.
Then any (randomized) statistical algorithm that decides correctly whether $D \in \D_Q$ or $D=U_k$ with probability at least $2/3$ (over the choice of $D$ and randomness of the algorithm) needs either
\begin{enumerate}
\item $m$ calls to the $\MSAMPLE(L)$ oracle with
$
m\cdot L \ge c_1 \left(\frac{n}{\log n}\right)^{r}
$
for a constant $c_1 = \Omega_k(1)$, OR
\item $q$ queries to $\MVSTAT\left(L,\frac{c_2}{L}\cdot \frac{n^r}{(\log q)^{r}}\right)$ for a constant $c_2 = \Omega_k(1)$ and any $q \geq L$.
\end{enumerate}
\end{thm}
The first part of the theorem exhibits the trade-off between the number of queries $m$ and the number of values the query can take $L$. It might be helpful to think of the latter as evaluating $L$ disjoint functions on a random sample, a task that would have complexity growing with $L$. The second part of the theorem is a superpolynomial lower bound (in $n$, for any fixed $r$) if the parameter $t$ (recall the oracle is allowed only an error equal to the standard deviation of averaging over $t$ random samples) is less than $n^{r}/(\log n)^{2r}$.

\subsection{Algorithms}
We next turn to our algorithmic results, motivated by two considerations. First, the $O(n^{r/2})$-clause algorithm implicit in \cite{bogdanov2009security} does not appear to lead to a non-trivial statistical algorithm. Second, much of the literature on upper bounds for planted problems uses spectral methods, and so we aim to implement such spectral algorithms statistically.

The algorithm we present is statistical and nearly matches the lower bound. It can be viewed as a discrete rounding of the power iteration algorithm for a suitable matrix constructed from the clauses of the input.
\begin{thm}\label{thm:algo1}
Let $\Z_Q$ be a planted satisfiability problem with clause distribution $Q$ having distribution complexity $r$. Then there exists an algorithm to solve  $\Z_Q$  using $O(n^{r/2}\log n) $ random clauses and time linear in this number.  This algorithm can be implemented statistically in any of the following ways.
\begin{enumerate}
\item Using $O(n^{r/2} \log^2 n)$ calls to $\MSAMPLE(n^{\lceil r/2\rceil})$;
\item For even $r$: using $O(\log n)$ calls to $\MVSTAT(n^{r/2},n^{r/2} \log \log n)$;
\item For odd $r$: using $O(\log n)$ calls to $\MVSTAT(O(n^{\lceil r/2\rceil}), O(n^{r/2}\log n))$;
\end{enumerate}
\end{thm}
Thus for any $r$, the upper bound matches the lower bound up to logarithmic factors for sample size parameter $t=n^{r/2}$,  with $L = n^{\lceil r/2\rceil}$ being only slightly higher in the odd case than the $L=n^{r/2}$ that the lower bound implies for such $t$.
The algorithm is a discretized variant of the algorithm based on power iteration with subsampling from \cite{FeldmanPV14}.  The upper bound holds for the problem of finding the planted assignment exactly, except in the case $r=1$.  Here  $\Omega(n \log n)$ clauses are required for complete identification since that many clauses are needed for each variable to appear at least once in the formula. In this case $O(n^{1/2} \log n)$ samples suffice to find an assignment with non-trivial correlation with the planted assignment, i.e. one that agrees with the planted assignment on $n/2 + t \sqrt n$ variables for an arbitrary constant $t$.

\subsection{Statistical dimension for decision problems}
\label{sec:statdim}
For a domain $X$, let $\D$ be a set of distributions over $X$ and let $D$ be a distribution over $X$ which is not in $\D$. 
For $t>0$, the {\em distributional decision problem} $\B(\D,D)$ using $t$ samples is to decide, given access to $t$ random samples from an arbitrary unknown distribution $D' \in \D \cup \{D\}$, whether $D' \in \D$ or $D' = D$. Lower bounds on the complexity of statistical algorithms use the notion of {\em statistical dimension} introduced in \cite{FeldmanGRVX:12}, based on ideas from \cite{BlumFJ+:94,Feldman:12jcss}.

To prove our bounds we introduce a new, stronger notion of statistical dimension which directly examines a certain norm of the operator that discriminates between expectations taken relative to different distributions.  Formally, for a distribution $D' \in \D$ and a reference distribution $D$ we examine the (linear) operator that maps a function $h:X\rightarrow \R$ to $\E_{D'}[h]-\E_{D}[h]$. Our goal is to obtain bounds on a certain norm of this operator extended to a set of distributions. Specifically, the {\em discrimination norm} of a set of distributions $\D'$ relative to a distribution $D$ is denoted by $\dc(\D',D)$ and defined as follows:
\begin{align*}
\dc(\D',D) \doteq \max_{h, \|h\|_D=1} \left\{ \E_{D' \sim \D'}\left[\left| \E_{D'}[h]-\E_{D}[h]\right| \right] \right\},
\end{align*}
where the norm of $h$ over $D$ is $\|h\|_D = \sqrt{\E_D[h^2(x)]}$ and $D' \sim \D'$ refers to choosing $D'$ randomly and uniformly from the set $\D'$. A statistical algorithm can get an estimate of the expectation of a query $h$ and use the value to determine whether the input distribution is the reference distribution or one of the distributions in $\D'$. Intuitively, the norm measures how well this can be done on average over distributions in $\D'$. In particular, we will show that if $\dc(\D',D) = \kappa$ then a single query to $\VSTAT(1/(3\kappa^2))$ cannot be used to distinguish all distributions in $\D'$ from $D$.

Our concept of statistical dimension is essentially the same as in \cite{FeldmanGRVX:12} but uses $\dc(\D',D)$ instead of average correlations.
\begin{definition}\label{def:sdima}
  For $\kappa>0$, domain $X$ and a decision problem $\B(\D,D)$, let $d$ be the largest integer
  such that there exists a finite set of distributions $\D_D \subseteq \D$ with the following property:
  for any subset $\D' \subseteq \D_D$, where $|\D'| \ge |\D_D|/d$, $\dc(\D',D) \leq \kappa$.
The \textbf{statistical dimension} with discrimination norm $\kappa$ of $\B(\D,D)$
is $d$ and denoted by $\SDN(\B(\D,D),\kappa)$.
\end{definition}
The dimension is equal to (at least) $d$ if there exists a reference distribution $D$ and a ``hard" subset of distributions $\D_D$, such that no large subset of $\D_D$ has discrimination norm larger than $\kappa$ (and, consequently, cannot be distinguished from $D$ using a single query to $\VSTAT(1/(3\kappa^2))$). Here large subset means at least $1/d$ fraction of distributions in $\D_D$. We remark that this statistical dimension can be easily extended to general search problems as in \cite{FeldmanGRVX:12} (the extension can be found in an earlier version of this work \cite[v5]{FeldmanPV:13}). A detailed treatment and additional approaches to proving statistical query lower bounds for search problems can be found in a subsequent work of Feldman \cite{Feldman:16sqd}.

The statistical dimension with discrimination norm $\kappa$ of a problem over distributions gives
a lower bound on the complexity of any statistical algorithm.
\begin{thm}
\label{thm:avgvstat-random}
  Let $X$ be a domain and $\B(\D,D)$ be a decision problem over a class of distributions $\D$ on $X$ and reference distribution $D$. For $\kappa > 0$, let $d = \SDN(\B(\D,D),\kappa)$ and let $L \geq 2$ be an integer.
  \begin{itemize}
  \item  Any randomized statistical algorithm that solves $\B(\D,D)$ with probability $\geq 2/3$ over the randomness in the algorithm requires $\Omega(d/L)$ calls to $\MVSTAT(L,1/(12 \cdot \kappa^2 \cdot L))$.
  \item   Any randomized statistical algorithm that solves $\B(\D,D)$ with probability $\geq2/3$ over the randomness in the algorithm requires at least $m$ calls to $\MSAMPLE(L)$ for $m = \Omega\left(\min\left\{d,1/\kappa^2 \right\}/L\right)$.
  \end{itemize}
  Further, the lower bound also holds when the input distribution $D'$ is chosen randomly as follows: $D'=D$ with probability $1/2$ and $D'$ equals  a random and uniform element of $\D_D$ with probability $1/2$, where $\D_D$ is the set of distributions for which the value of $d$ is attained.
\end{thm}
We prove this theorem in a slightly more general form in Section \ref{sec:general-stat-bounds}. Our proof relies on techniques from \cite{FeldmanGRVX:12} and simulations of $\MVSTAT$ and $\MSAMPLE$ using $\VSTAT$ and $\SAMPLE$, respectively.

In our setting the domain $X_k$ is the set of all clauses of $k$ ordered literals (without variable repetition); the class of distributions $\D_Q$ is the set of all distributions $Q_\sigma$ where $\sigma$ ranges over all $2^n$ assignments; the distribution $D$ is the uniform distribution over $X_k$ referred to as $U_k$.

In the Section \ref{sec:plantstat-sdn} we prove the following bound on the statistical dimension with discrimination norm of planted satisfiability.
\begin{thm}
\label{thm:main-sdn-bound}
 For any distribution $Q$ over $k$-clauses of distributional complexity $r$, there exists a constant $c > 0$ (that depends on $Q$) such that for any $q \geq 1$,
 $$\SDN\left(\B(\D_Q,U_k), \frac{c (\log q)^{r/2}}{n^{r/2}}\right) \geq q.$$
\end{thm}

For an appropriate choice of $q=n^{\theta(\log n)}$ we get, $\SDN(\B(\D_Q,U_k), \frac{(\log n)^{r}}{n^{r/2}}) = n^{\Omega_k(\log n)}.$ Similarly, for any constant $\eps>0$,  we get $\SDN(\B(\D_Q,U_k), n^{r/2-\eps}) = 2^{n^{\Omega_k(1)}}.$ By using this bound in Theorem \ref{thm:avgvstat-random} we obtain our main lower bounds in Theorem \ref{thm:lower-bound}.

In Section \ref{sec:csp-extend} we prove the same lower bound for the generalized planted $k$-CSP problem. Our proof is based on a reduction showing that any statistical query for an instance of the $k$-CSP problem of complexity $r$ can be converted to a query for a planted $k$-SAT instance of distribution complexity $r$. The reduction ensures that the resulting query is essentially as informative in distinguishing the planted distribution from the reference one as the original query. As a result it reduces a bound on $\dc$ of a planted $k$-CSP problem to an almost equivalent bound on $\dc$ of the corresponding planted $k$-SAT problem.

\subsection{Corollaries and applications}
\label{sec:hardSAT}

\subsubsection{Quiet plantings}

Finding distributions of planted $k$-SAT instances that are algorithmically intractable has been a pursuit of researchers in both computer science and physics.  It was recognized in \cite{barthel2002hiding,jia2005generating} that uniform planted $k$-SAT is easy algorithmically due to the bias towards true literals, and so they proposed distributions in which true and false literals under the planted assignment appear in equal proportion. Such distributions have complexity $r \ge 2$ in our terminology.  These distributions have been termed `quiet plantings' since evidence of the planting is suppressed.

Further refinement of the analysis of quiet plantings was given in \cite{krzakala2012reweighted}, in which the authors analyze belief propagation equations and give predicted densities at which quiet plantings transition from intractable to tractable.  Their criteria for a quiet planting is exactly the equation that characterizes distribution complexity $r \ge 2$, and the conditions under which the tractability density diverges to infinity corresponds to distribution complexity $r \ge 3$.

The distribution complexity parameter defined here generalizes quiet plantings to an entire hierarchy of quietness.  In particular, there are distributions of satisfiable $k$-SAT instances with distribution complexity as high as $k-1$ ($r=k$ can be achieved using XOR constraints but these instances are solvable by Gaussian elimination). Our main results show that for distributions with complexity  $ r\ge 3$, the number of clauses required to recover the planted assignment is super-linear (for statistical algorithms with $L \le n^{r/2}$).

For examples of such distributions, consider weighting functions $Q(y)$ that depend only on the number of true literals in a clause under the planted assignment $\sigma$.  We will write $Q(j)$ for the value of $Q$ on any clause with exactly $j$ true literals.  Then setting $Q(0) =0, Q(1)=3/32, Q(2) =1/16, Q(3)=1/32, Q(4)=1/8$ gives a distribution over satisfiable $4$-SAT instances with distribution complexity $r=3$, and an algorithmic threshold at  $\tilde \Theta(n^{3/2})$ clauses.
Similar constructions for higher $k$ yield distributions of increasing complexity with algorithmic thresholds as high as $\tilde \Theta(n^{(k-1)/2})$.  These instances are the most `quiet' proposed and can serve as strong tests of industrial SAT solvers as well as the underlying hard instances in cryptographic applications. Note that in these applications it is important that a hard SAT instance can be obtained from an easy to sample planted assignment $\sigma$. Our lower bounds apply to the uniformly chosen $\sigma$ and therefore satisfy this condition.

\subsubsection{Feige's Hypothesis}

As a second application of our main result, we show that Feige's $3$-SAT hypothesis \cite{feige2002relations} holds for the class of statistical algorithms.  A refutation algorithm takes a $k$-SAT formula $\Phi$ as an input and returns either SAT or UNSAT. The algorithm must satisfy the following:
\begin{enumerate}
\item If $\Phi$ is satisfiable, the algorithm always returns SAT.
\item If $\Phi$ is drawn uniformly at random from all $k$-SAT formulas of $n$ variables and $m$ clauses, where $m/n$ is above the satisfiability threshold (the clause density at which the formula become unsatisfiable with high probability), then the algorithm must return UNSAT with probability at least $2/3$ (or some other arbitrary constant).
\end{enumerate}

 As with planted satisfiability, the larger $m$ is the easier refutation becomes, and so the challenge becomes finding efficient refutation algorithms that succeed on the sparsest possible instances. Efficient $3$-SAT refutation algorithms are known for $m =  \Omega (n^{3/2})$ \cite{coja2004strong, feige2004easily}.  Feige hypothesized 1) that no polynomial-time algorithm can refute formulas with $m \le \Delta n$ clauses for any constant $\Delta$ and 2) for every $\eps >0$ and large enough constant $\Delta$, there is no polynomial-time algorithm that answers UNSAT on most 3-SAT formulas but answers SAT on all formulas that have assignments satisfying $(1-\eps)$-fraction of constraints. Hypothesis 2 is strictly weaker than hypothesis 1.  Based on these hypotheses he derived hardness-of-approximation results for several fundamental combinatorial optimization problems.

To apply our bounds we need to first define a distributional version of the problem.
\begin{defn}
\label{def:refute-distr}
In the distributional $k$-SAT refutation problem the input formula is obtained by sampling $m$ i.i.d. clauses from some unknown distribution $D$ over clauses. An algorithm successfully solves the distributional problem if:
\begin{enumerate}
\item The algorithm returns SAT for every distribution supported on simultaneously satisfiable clauses.
\item The algorithm returns UNSAT with probability at least $2/3$ when clauses are sampled from the uniform distribution and $m/n$ is above the satisfiability threshold.
\end{enumerate}
\end{defn}

\begin{prop}
The original refutation problem and distributional refutation problem are equivalent: a refutation algorithm for the original problem solves the distributional version and vice versa.
\end{prop}

\begin{proof}
The first direction is immediate: assume that we have a refutation algorithm $A$ for a fixed formula.
We run the refutation algorithm on the $m$ clauses sampled from the input distribution and output the algorithm's answer. By definition, if the input distribution is uniform then the sampled clauses will give a random formula from this distribution. So $A$ will return UNSAT with probability at least $2/3$. If the clauses in the support of the input distribution can be satisfied then the formula sampled from it will be necessarily satisfiable and $A$ must return SAT.

In the other direction, we again run the distributional refutation algorithm $A$ on the $m$ clauses of $\Phi$ and output its answer (each clause is used as a new sample consecutively). If $\Phi$ was sampled from the uniform distribution above the satisfiability threshold, then the samples we produced are distributed according to the uniform distribution. Therefore, with probability at least $2/3$ $A$ returns UNSAT. If $\Phi$ is satisfiable then consider the distribution $D_\Phi$ which is uniform over the $m$ clauses of $\Phi$. $\Phi$ has non-zero probability to be the outcome of $m$ i.i.d. clauses sampled from $D_\Phi$. Therefore $A$ must output SAT on it since otherwise it would violate its guarantees. Therefore the output of our algorithm will be SAT for $\Phi$.
\end{proof}

In the distributional setting, an immediate consequence of Theorem \ref{thm:lower-bound} is that Feige's hypothesis holds for the class of statistical algorithms.
\begin{thm}
\label{cor:feige}
Any (randomized) statistical algorithm that solves the distributional $k$-SAT refutation problem requires: 
\begin{enumerate}
\item $m$ calls to the $\MSAMPLE(L)$ oracle with $m\cdot L \ge c_1 \left(\frac{n}{\log n}\right)^{k}$
for a constant $c_1 = \Omega_k(1)$.
\item $q$ queries to $\MVSTAT\left(L,\frac{c_2}{L}\cdot \frac{n^k}{(\log q)^{k}}\right)$ for a constant $c_2 = \Omega_k(1)$ and any $q \geq L$.
\end{enumerate}
\end{thm}

 \begin{proof}
 The decision problem in Theorem \ref{thm:lower-bound} is a special case of the distributional refutation problem. Specifically, say there is such a refutation algorithm.
Let $Q$ be a fully satisfiable clause distribution with distribution complexity $k$. Then consider a distribution $D$ so that  either $D = U_k$ or $D= Q_\sigma \in \D_Q$ for a uniformly chosen $\sigma \in \{ \pm 1 \}^n$. Then run the refutation algorithm on $D$.  If $D \in \D_Q$, then the algorithm must output SAT, and so we conclude $D \in \D_Q$.  If $D = U_k$, then with probability $2/3$ the algorithm must output UNSAT in which case we conclude that $D = U_k$.  This gives an algorithm for distinguishing $U_k$ from $\D_Q$ with probability at least $2/3$, a contradiction to Theorem \ref{thm:lower-bound}.
\end{proof}

If $r \geq 3$ and $L \le n^{r/2}$, the lower bound on the number of clauses $m$ is $\tilde{\Omega}(n^{r/2})$ and is much stronger than $\Delta n$ conjectured by Feige. Such a stronger bound is useful for some hardness of learning results based on Feige's conjecture \cite{DanielyLS:13}. For $k=3$, the $\tilde{\Omega}(n^{3/2})$ lower bound on $m$ essentially matches the known upper bounds \cite{coja2004strong, feige2004easily}.

We note that the only distributions with $r=k$ are noisy $k$-XOR-SAT distributions. Such distributions generate satisfiable formulas only when the noise rate is 0 and then formulas are refutable via Gaussian elimination. Therefore if one excludes the easy (noiseless) $k$-XOR-SAT distribution then we obtain only the stronger form of Feige's conjecture ($\eps > 0$) with $r=k=3$.

\subsubsection{Hardness of approximation}
We note finally that optimal inapproximability results can be derived from Theorem \ref{thm:lower-bound} as well, including the fact that pairwise independent predicates (as studied in \cite{austrin2009approximation}) are approximation-resistant for the class of statistical algorithms.

Our work provides a means to generate candidate distributions of hard instances for approximation algorithms for CSP's: find a distribution $Q$ on $\{ \pm 1 \}^k$ supported only on vectors that satisfy the CSP predicate with high distribution complexity (as in the example of $4$-SAT above).  Then statistical algorithms cannot efficiently distinguish the planted distribution (all constraints satisfied) from the uniformly random distribution (eg. ($1-2^{-k}$)-fraction of constraints satisfied in the case of $k$-SAT).

\section{Convex Programs and SQ Algorithms for Solving CSPs}
\label{sec:statalgpowersec}
\newcommand{\sqhsdv}{{\tt HS-DV}}

In this section we show how our lower bounds for planted $k$-SAT together with general statistical query algorithms for solving stochastic convex programs from \cite{FeldmanGV:15}  imply lower bounds on convex programs that can be used to solve planted $k$-SAT (analogous results also hold for Goldreich's $k$-CSP but we omit them for brevity). At a high level we observe that a convex relaxation can be viewed as a reduction from our planted constraint satisfaction problem to a stochastic convex optimization problem. Existence of such a reduction together with a statistical query algorithm for the corresponding stochastic convex program would violate the lower bounds that we prove. Hence, as a contrapositive, we rule out existence of several types of convex relaxations for the planted CSPs.

\subsection{LP/SDP relaxations for $k$-CSPs}
We first describe several standard ways in which Boolean constraint satisfaction problems\footnote{As usual in this literature, constraint satisfaction also refers to the problem of maximizing  the number of satisfied constraints.} are relaxed to an LP or an SDP.

The classic SDP for a constraint satisfaction problem is the MAX-CUT SDP of Goemans and Williamson \cite{GoemansW95}. In this program the goal is to maximize $\sum_{i,j \in [n]}[e_{ij} (1-x_{i,j})]$, where
$e_{ij} \in \R$ is the indicator of an edge presence in the graph and $x$, viewed as an $n\times n$ matrix, is constrained to be in the PSD cone with some normalization.

More generally, the canonical LP relaxation of a $k$-CSP with $m$ constraints results in a program of the following type (see \cite{ODonnell:11course} for a textbook version or \cite{Raghavendra:08,barak2013optimality,o2013goldreich} for some applications):
$$\mbox{maximize } \sum_{i \in [m]}\left( \sum_{y\in \{\pm 1\}^k, y \mbox{ satisfies } C_i } x_{V_i,y}\right) ,$$ subject to $\bar{x} \in K$. Here, $C_i$ denotes the $k$-ary Boolean predicate of the $i$-th constraint, $V_i$ denotes the $k$-tuple of variables of $i$-th constraint and $x_{V_i,y}$ is the variable that tells whether variables in $V_i$ are assigned values $y$ (its constrained to be in $[0,1]$ and interpreted as probability). The set $K$ is an $O_k(n^k)$-dimensional convex set that makes sure that $x_{V_i,y}$'s are consistent in a natural sense (with additional PSD cone constraints in the case of SDPs).

Such relaxations are a special case of even more general relaxations that are based on a linearization of the Boolean constraints (for example \cite{KothariMR17}). Specifically, a linearization maps each assignment $x \in \{\pm 1\}^n$ to a vector $v_x \in \R^N$ and each relevant $k$-ary Boolean constraint $C$ to a vector $w_C\in \R^N$. 
For a convex body $K$ in $\R^N$ that includes $\{v_x \cond x \in \{\pm 1\}^n \}$, given a set $C_1,\ldots,C_m$ of Boolean predicates on $x$ one considers the program $\max_{v \in K} \sum_{i\in [m]} \la v,w_{C_i}\ra$. Note that condition  $\{v_x \cond x \in \{\pm 1\}^n \} \subset K$ ensures that the optimum of this program is always at least as large as the optimum of the original problem.
In order to be useful, such a relaxation also needs to be able to distinguish between instances with high and low values of the optimum in the original program (for some appropriate values of ``high" and ``low"). 

Here we consider an even more general class of convex relaxations. First, we allow mapping Boolean constraints to general convex objective functions. That is, a Boolean function $C$ over $\on^n$ is mapped to a convex function $f_C$ over a convex body $K \in \R^N$. We will not need an explicit mapping between the Boolean input and vectors in $K$ but will only require that the value of the optimum of the resulting objectives allows us to distinguish between instances with high and low values of the optimum in the original program. Also note that going beyond linear objectives implies that we will be minimizing (and not maximizing) the resulting objective.

Our lower bounds apply to distributional CSPs in which constraints are sampled i.i.d.~from some distribution $D$ and they show that even estimating the value of the expected objective: $\max_{\sigma \in \on^n} \E_{C\sim D}[ \sigma(C)]$ is hard. We remark that, given $m = \Omega(n/\eps^2)$ samples, for every $x \in \on^n$, the value of the objective based on $m$ i.i.d. samples is within $\eps$ of the value of the expected objective (with high probability). Applying a convex relaxation to such a CSP leads to the following convex program: $\min_{x \in K} \E_{C \sim D} [f_C(x)]$, where $K$ is a fixed convex $N$-dimensional set (that is not dependent the distribution $D$) and for every $C$, $f_C(x)$ is a bounded convex function over $K$.  Such programs are referred to as {\em stochastic convex programs} and are well-studied in machine learning and optimization (e.g.~\cite{NemirovskiJLS09,Shalev-ShwartzSSS09}).

\newcommand{\Opt}{{\mbox{Opt}}}
\subsection{Statistical Query Algorithms for Stochastic Convex Optimization}
We now describe several results from \cite{FeldmanGV:15} giving upper bounds on solving various stochastic convex programs by statistical query algorithms. Bounds for a number of additional types of convex programs are given in \cite{FeldmanGV:15} and can be applied in this context in a similar way. We start by defining the problem of distribution-independent stochastic convex optimization formally.
\begin{definition}
For a convex set $K$, a set $\F$ of convex functions over $K$ and $\eps>0$ we denote by $\Opt(K,\F,\eps)$ the problem of finding, for every distribution $D$ over $\F$, $x^*$ such that $f_D(x^*) \leq \min_{x\in K} f_D(x) + \eps$, where $f_D(x) \doteq \E_{f \sim D}[f(x)]$.
\end{definition}

\paragraph{Center-of-gravity:}
For general convex functions with range scaled to $[-1,1]$, Feldman \etal \cite{FeldmanGV:15} describe two statistical query algorithms that both use $\VSTAT(O(N^2/\eps^2))$ to find an $\eps$-approximate solution to the stochastic convex program. The first algorithm is  based on the random walk approach from \cite{KV2006, LV2006b}. The second algorithm is based on  the classic center-of-gravity method \cite{Levin:1965} and requires fewer queries.
\begin{thm}
\label{thm:cog}[\cite{FeldmanGV:15}]
Let $K\subseteq \R^N$ be a convex body and let $\F$ be the set of all convex functions over $K$ such that for all $x\in K, |f(x)| \leq 1$. Then
 there is an algorithm that solves $\Opt(K,\F,\eps)$ using $O(N^2\log(1/\eps))$ queries to $\VSTAT(O(N^2/\eps^2))$.
\end{thm}
The theorem above ignores computational considerations since those do not play any role in our information-theoretic lower bounds. An efficient version of this algorithm is also given in \cite{FeldmanGV:15}.

\paragraph{Mirror descent:}
In most practical cases the expected convex objective is optimized using simpler methods such as gradient-descent based algorithms. It is easy to see that such methods fit naturally into the statistical query framework. For example, gradient descent relies solely on knowing $\nabla f_D(x_t)$ approximately, where $f_D$ is the optimized function and $x_t$ is the solution at step $t$. By linearity of expectation, we know that $\nabla \E_D[f(x_t)] = \E_D[\nabla f(x_t)]$. This means that we can approximate $\nabla \E_D[f(x_t)]$ using queries to $\VSTAT$ with sufficiently large parameter. In particular, as shown in \cite{FeldmanGV:15}, the classic mirror-descent method \cite{nemirovsky1983problem} can be implemented using a polynomial number of queries to $\VSTAT(O((L \cdot R \cdot \log(N)/\eps)^2))$ to $\eps$-approximately solve any convex program whenever $K$ is contained in the $\ell_p$ (for any $p \in [1,2]$) ball of radius $R$ and the $\ell_q$ (for $q = 1/(1-1/p)$) norm of $\nabla f$ is bounded by $L$. Note that in this case the dependence of the accuracy parameter of $\VSTAT$ on the dimension is just logarithmic.  We denote $\B_p^N(R) \doteq \{x \cond \|x\|_p \leq R\}$.

\begin{thm}
\label{thm:solve_cvx_ellp}
Let $p\in[1,2]$, $L,R>0$, and $K \subseteq \B_p^N(R)$ be a convex body. Let $\F$ be the set of all functions $f$ that satisfy, for all $x\in K$, $\|\nabla f(x)\|_q \leq L$ for $q=1-(1/p)$. Then there is an algorithm that solves $\Opt(K,\F,\eps)$ using
$O\left(N\log N\cdot (LR/\eps)^2\right)$ queries to $\VSTAT(O((\log N \cdot LR/\eps)^2))$.
\end{thm}

\newcommand{\Stat}{{\mbox{Stat}}}
\newcommand{\cO}{{\mathcal O}}
\subsection{Corollaries for Planted $k$-CSPs}

Now observe that a convex relaxation (of the type that we defined) is just a mapping from Boolean constraints to convex functions in some class of functions $\F$ over a convex body $K$. Such mapping allows to implement a statistical query oracle for the distribution over convex functions given a statistical query oracle for the input distribution over Boolean constraints. In particular, it allows us to run a statistical query algorithm for $\Opt(K,\F,\eps)$ on the stochastic convex program that corresponds to the input distribution over $k$-clauses.
Now assume that the value of the solution for the stochastic convex program corresponding to a planted $k$-CSP is smaller than the value of the solution for the the stochastic convex program corresponding to the uniform distribution over constraints by at least $\eps$. Then a statistical query algorithm for $\Opt(K,\F,\eps)$ solves the decision version of our planted $k$-CSP. Hence, if $\Opt(K,\F,\eps)$ can be solved using some number of queries to a statistical oracle that violates our lower bound then a convex relaxation that satisfies these properties cannot exist.
We now make these statements formally.

\begin{thm}\label{thm:lower-bound-convex}
Let $Q$ be a distribution over $k$-clauses of complexity $r$. Assume that there exists a mapping that maps each $k$-clause $C\in X_k$ to a convex function $f_C:K \rightarrow[-1,1]$ over some convex $N$-dimensional set $K$ that for some $\eps > 0$ and $\alpha$ satisfies:
\begin{itemize}
\item $\Pr_{\sigma \in \on^n}\left[ \min_{x\in K}\left\{ \E_{C \sim Q_\sigma}[f_C(x)] \right\} \leq \alpha\right] \geq 1/2$;
\item $\min_{x\in K} \{\E_{C \sim U_k} [f_C(x)] \} > \alpha +\eps$.
\end{itemize}
Then for every $q \geq 1$,  solving $\Opt(K,\F,\eps)$ using $\VSTAT\left(\frac{n^{r}}{(\log q)^{r}}\right)$ requires $\Omega(q)$ queries.
\end{thm}

The first condition on the value of the solution requires that for most $\sigma$'s the value of the minimum of the expected objective is at most $\alpha$. The planted instances are usually the instances which have a higher number of satisfied clauses so this is a weakening of the condition that a convex relaxation should only decrease the value of the minimum. The second condition requires that the value of the minimum of the expected objective on the uniform distribution is at least $\alpha+\eps$. For comparison, the standard condition on a convex relaxation requires that the value of the solution obtained on $m$ clauses randomly sampled from the uniform distribution is large. Our condition is weaker since for every $m$ and $x^* = \argmin_{x\in K} \{\E_{C \sim U_k} [f_C(x)]\}$,
$$ \E_{C_1,\ldots,C_m\sim U_k}\left[\min_{x\in K} \left\{\frac{1}{m} \sum_i f_{C_i}(x) \right\}\right] \leq \E_{C_1,\ldots,C_m\sim U_k}\left[\frac{1}{m} \sum_i f_{C_i}(x^*) \right] = \min_{x\in K} \{\E_{C \sim U_k} [f_C(x)] \}.$$

Now Theorem \ref{thm:lower-bound-convex} can be combined with upper bounds on the complexity of solving stochastic convex programs we gave above to obtain lower bounds on the parameters of convex relaxations that can be used to solve planted satisfiability problems. For example the algorithm described in Theorem \ref{thm:cog} implies Corollary \ref{cor:lower-convex-program-intro}.
Using Theorem \ref{thm:solve_cvx_ellp} we can exclude convex relaxations even in exponentially high dimension as long as the convex set is bounded in $\ell_p$ norm and functions satisfy a Lipschitz condition. For simplicity we take these constraints to be 1 (a more general statement can be obtained easily by rescaling).
\begin{cor}
\label{cor:lower-convex-program-norm}
Let $Q$ be a distribution over $k$-clauses of complexity $r$. For $p\in [1,2]$, let $K\subseteq \B_p^N(1)$ be convex and compact set and $\F_p = \left\{f(\cdot) \cond \forall x \in K, \|\nabla f(x)\|_q \leq 1\right\}$ (where $q = 1/(1-1/p)$). Assume that there exists a mapping that maps each $k$-clause $C\in X_k$ to a convex function $f_C \in \F_p$. Further assume that for some $\eps > 0$ and $\alpha \in \R$,
 $$\Pr_{\sigma \in \on^n}\left[ \min_{x\in K}\left\{ \E_{C\sim Q_\sigma}[ f_{C}(x)]  \right\} \leq \alpha\right] \geq 1/2.$$
and
$$\min_{x\in K} \left\{\E_{C\sim U_k}[f_{C}(x)] \right\} > \alpha + \eps .$$
Then $N = 2^{\tilde{\Omega}_k\left(n^{r/(r+2)}  \cdot \eps^{2/(r+2)}\right)}$.
\end{cor}

For instances of planted satisfiability there is a constant gap between the fraction of clauses that can be satisfied in a formula sampled from $U_k$ and the fraction of clauses that can be satisfied in a formula sampled from $Q_\sigma$. Thus, for convex relaxations that satisfy the conditions of Corollary \ref{cor:lower-convex-program-norm} the lower bounds imply a large integrality gap.

Note that these corollaries give concrete lower bounds on the dimension and other structural properties of convex programs that can be used to solve an average-case $k$-CSP without any assumptions about how the convex program is solved. In particular, it does not need to be solved via a statistical algorithm or even computationally efficiently. As far as we know, this approach to obtaining lower bounds for convex relaxations from convex optimization algorithms and statistical query lower bounds is new.

\begin{rem}
We observe that standard lift-and-project procedures (Sherali-Adams, Lov\'asz-Schrijver, Lasserre) for strengthening LP/SDP formulations do not affect the analysis above. While these procedures add a large number of auxiliary variables and constraints the resulting program is still a convex optimization problem in the same dimension (although implementation of the separation oracle becomes more computationally intensive). Hence the use of such procedures does not necessarily affect the bounds on the number of queries and tolerance we gave above.
\end{rem}

At a more conceptual level, the primary difference between the commonly considered hierarchies of LP/SDP relaxations and our approach is as follows. The expected objective value of the stochastic convex programs corresponding to these hierarchies of relaxations captures the expected objective of the original Boolean $k$-CSP. Yet, solving stochastic convex programs corresponding to these relaxations for all distributions requires $\Omega(n^{r/2})$ samples information theoretically (e.g.~\cite{FeldmanGV:15}). Lower bounds against such relaxations effectively prove that this number of samples is necessary even for the uniform distribution over the clauses: given fewer samples the optimum of the objective based on the given random samples will have a much lower value than the optimum of the expected objective (a phenomenon that is referred to as overfitting). In contrast, our approach rules out relaxations for which the resulting stochastic convex program can be solved by a statistical query algorithm using $q$ queries to $\VSTAT\left(\frac{n^{r}}{(\log q)^{r}}\right)$. In particular there is no overfitting. However,  such relaxations end up not being sufficiently expressive: the optimum of the expected objective of the relaxation does not differentiate between the planted distributions and the uniform one. This difference makes our lower bounds incomparable and, in a way, complementary to existing work on lower bounds for specific hierarchies of convex relaxations.

\section{Statistical Dimension of Planted Satisfiability}
\label{sec:plantstat-sdn}
In this section, we prove our lower bound on the statistical dimension with discrimination norm of the planted satisfiability problem (Theorem \ref{thm:main-sdn-bound}). Recall that the theorem states that for any distribution $Q$ over $k$-clauses of distributional complexity $r$, there exists a constant $c > 0$  such that for any $q \geq 1$,
 $$\SDN\left(\B(\D_Q,U_k), \frac{c (\log q)^{r/2}}{n^{r/2}}\right) \geq q.$$

\noindent {\bf Proof overview:} We first show that the discrimination operator corresponding to $Q$ applied to a function $h:X_k\rightarrow \R$ can be decomposed into a linear combination of discrimination operators for $\ell$-XOR-SAT. Namely, we show that
$$\E_{Q_\sigma}[h] - \E_{U_k}[h] =  -2^k\sum_{S \subseteq [k]} \hat{Q}(S) \cdot (\E_{Z_{\ell,\sigma}}[h_S] - \E_{U_\ell}[h_S]),$$
where $Z_{\ell,\sigma}$ is the $\ell$-XOR-SAT distribution over $\ell$-clauses with planted assignment $\sigma$, and $h_S$ is a projection of $h$ to $X_\ell$ defined below.

The two key properties of this decomposition are: $(i)$ the coefficients obtained in the decomposition are exactly $\hat{Q}(S)$'s, which determine the distribution complexity of $Q$, and $(ii)$  $\|h_S\|_{U_\ell}$ is upper-bounded by $\|h\|_{U_k}$. This step implies that the discrimination norm for the problem defined by $Q$ is upper bounded (up to constant factors) by the discrimination norm for $r(Q)$-XOR-SAT.

In the second step of the proof we bound the discrimination norm for the $r(Q)$-XOR-SAT problem. Our analysis is based on the observation that $\E_{Z_{\ell,\sigma}}[h_S] - \E_{U_\ell}[h_S]$ is a degree-$\ell$ polynomial as a function of $\sigma$. We exploit known concentration properties of degree-$\ell$ polynomials to show that the function cannot have high expectation over a large subset of assignments. This gives the desired bound on the discrimination norm for the $r(Q)$-XOR-SAT problem.

We now give the formal details of the proof.
For a distribution $Q_\sigma$ and query function $h:X_k \rightarrow \R$, we define $\Delta(\sigma,h) = \E_{Q_\sigma}[h] - \E_{U_k}[h]$. We start by introducing some notation:
\begin{definition}
For $\ell \in [k]$,
\begin{itemize}
\item Let $Z_\ell$ be the $\ell$-XOR-SAT distribution over $\on^\ell$, that is a distribution such that $Z_\ell(i)=1/2^{\ell-1}$ if $i$ is odd and 0 otherwise.
\item For a clause $C \in X_k$ and $S \subseteq [k]$ of size $\ell$, let $C_{|S}$ denote a clause in $X_\ell$ consisting of literals of $C$ at positions with indices in $S$ (in the order of indices in $S$).
\item For $h:X_k \rightarrow \R$, $S \subseteq [k]$ of size $\ell$ and $C_\ell \in X_\ell$, let $$h_S(C_\ell) = \frac{|X_\ell|}{|X_k| }\sum_{C\in X_k,\ C_{|S} = C_{\ell}} h(C).$$
\item For $g:X_\ell \rightarrow \R$, let $\Gamma_\ell(\sigma,g) = \E_{Z_{\ell,\sigma}}[g] - \E_{U_\ell}[g]$.
\end{itemize}
\end{definition}

Recall the discrete Fourier expansion of a function $Q: \{ \pm 1 \}^k \to \R$:
\[ Q(x) = \sum_{S \subseteq [k]} \hat Q(S) \chi_S(x), \]
where $\chi_S(x) = \prod_{i \in S} x_i$ is a parity or Walsh basis function, and the Fourier coefficient of the set $S$ is defined as:
\[ \hat Q(S) = \frac{1}{2^k} \sum_{y \in \{\pm 1\}^k} Q(y) \chi_S(y) \]

We show that $\Delta(\sigma,h)$ (as a function of $h$) can be decomposed into a linear combination of $\Gamma_\ell(\sigma,h_S)$.
\begin{lem}
\label{lem:delta-decompose}
For every $\sigma$ in $\on^n$ and $h:X_k \rightarrow \R$,
$$\Delta(\sigma,h) = - 2^k \sum_{S \subseteq [k], S\neq \emptyset} \hat{Q}(S) \cdot \Gamma_\ell(\sigma,h_S) .$$
\end{lem}
\begin{proof}
Recall that for a clause $C$ we denote by $\sigma(C)$ the vector in $\on^k$ that gives evaluation of the literals in $C$ on $\sigma$ with $-1$ corresponding to TRUE and 1 to FALSE. Also by our definitions, $Q_\sigma(C) = \frac{2^k \cdot Q(\sigma(C))}{|X_k|}$.
Now, using $\ell$ to denote $|S|$,
\alequ{\E_{Q_\sigma}[h] &=  \sum_{C \in X_k} h(C) \cdot Q_\sigma(C) = \frac{2^k}{ |X_k|} \sum_{C \in X_k} h(C) \cdot Q(\sigma(C)) \nonumber \\
&= \frac{2^k}{ |X_k|} \sum_{S \subseteq [k]} \hat{Q}(S) \sum_{C \in X_k} \chi_S(\sigma(C)) \cdot h(C) \nonumber \\
&= \frac{2^k}{ |X_k|} \sum_{S \subseteq [k]} \hat{Q}(S)  \sum_{C_\ell \in X_\ell} \sum_{C \in X_k, C_{|S} = C_\ell} \chi_S(\sigma(C)) \cdot h(C) \label{eq:decomp}
  }

Note that if $C_{|S} = C_\ell$ then for $\ell\geq 1$,
$$\chi_S(\sigma(C))= \chi_{[\ell]}(\sigma(C_\ell)) = 1-2^\ell \cdot Z_\ell(\sigma(C_\ell)) \ $$ and for $\ell = 0$, $\chi_\emptyset(\sigma(C)) = 1$. Therefore, for $\ell \ge 1$ and $\ell =0$ respectively,
$$\sum_{C \in X_k, C_{|S} = C_\ell} \chi_S(\sigma(C)) \cdot h(C) = (1-2^\ell \cdot Z_\ell(\sigma(C_\ell))) \cdot \sum_{C \in X_k, C_{|S} = C_\ell}  h(C) \mbox{ and } $$ $$\frac{2^k}{ |X_k|} \sum_{C \in X_k} [\hat{Q}(\emptyset) h(C)] = 2^k \cdot \hat{Q}(\emptyset) \cdot \E_{U_k}[h(C)] = \E_{U_k}[h(C)] ,$$ where $\hat{Q}(\emptyset)=2^{-k}$ follows from $Q$ being a distribution over $\on^k$.
Plugging this into eq.(\ref{eq:decomp}) we obtain
\alequn{\Delta(\sigma,h) &= \E_{Q_\sigma}[h] - \E_{U_k}[h] \\ &=
\frac{2^k}{ |X_k|} \sum_{S \subseteq [k], S\neq \emptyset} \hat{Q}(S) \sum_{C_\ell \in X_\ell} \left[(1-2^\ell \cdot  Z_\ell(\sigma(C_\ell))) \cdot \sum_{C \in X_k, C_{|S} = C_\ell}  h(C) \right]\\
&=   \sum_{S \subseteq [k], S\neq \emptyset}\frac{2^k}{|X_\ell|} \hat{Q}(S) \sum_{C_\ell \in X_\ell} \left[(1-2^\ell \cdot Z_\ell(\sigma(C_\ell))) \cdot h_S(C_\ell) \right] \\
&= 2^k \sum_{S \subseteq [k], S\neq \emptyset} \hat{Q}(S) \left(\E_{U_{\ell}}[h_S] - \E_{Z_{\ell,\sigma}}[h_S]\right) \\
&= - 2^k \sum_{S \subseteq [k], S\neq \emptyset} \hat{Q}(S) \cdot \Gamma_\ell(\sigma,h_S),
}
where we used that, by definition of $Z_{\ell,\sigma}$, $\frac{1}{|X_\ell|}  
\cdot 2^\ell \cdot Z_\ell(\sigma(C_\ell)) = Z_{\ell,\sigma}(C_\ell)$.
\end{proof}

We now analyze $\Gamma_\ell(\sigma,h_S)$. For a clause $C$ let $I(C)$ denote the set of indices of variables in the clause $C$ and let $\overline\#(C)$ denote the number of negated variables is $C$. Then, by definition, $$Z_{\ell,\sigma}(C) = \frac{Z_\ell(\sigma(C))}{|X_\ell|} = \frac{1 - (-1)^{\overline\#(C)} \cdot \chi_{I(C)}(\sigma)}{|X_\ell|}. $$
This implies that $\Gamma_{\ell}(\sigma,h_S)$ can be represented as a linear combination of parities of length $\ell$.
\begin{lem}
\label{lem:gamma}
For $g:X_\ell \rightarrow \R$,
$$\Gamma_\ell(\sigma,g) = - \frac{1}{|X_\ell|} \sum_{A \subseteq [n], |A| = \ell} \left(\sum_{C_\ell \in X_\ell, I(C_\ell)=A} g(C_\ell) \cdot (-1)^{\overline\#(C_\ell)} \right) \cdot \chi_A(\sigma) .$$
\end{lem}
\begin{proof}
\alequn{
\Gamma_\ell(\sigma,g) &= \E_{Z_{\ell,\sigma}}[g] - \E_{U_\ell}[g] \\
& = -\frac{1}{|X_\ell|} \sum_{C_\ell \in X_\ell} g(C_\ell) \cdot (-1)^{\overline\#(C_\ell)} \cdot \chi_{I(C_\ell)}(\sigma) \\
&= - \frac{1}{|X_\ell|} \sum_{A \subseteq [n], |A| = \ell} \left(\sum_{C_\ell \in X_\ell, I(C_\ell)=A} g(C_\ell) \cdot (-1)^{\overline\#(C_\ell)} \right) \cdot \chi_A(\sigma).
}
\end{proof}
For $\cS \subseteq \on^n$ we now bound $\E_{\sigma \sim \cS}[|\Gamma_\ell(\sigma,g)|]$ by exploiting its concentration properties as a degree-$\ell$ polynomial. To do this, we will need the following concentration bound for polynomials on $\on^n$. It can be easily derived from the hypercontractivity results of Bonami and Beckner \cite{bonami1970etude,beckner1975inequalities} as done for example in \cite{Janson:97,DinurFKO:06}.
\begin{lem}
\label{lem:polynomial-concentration}
Let $p(x)$ be a degree $\ell$ polynomial over $\on^n$. Then there is constant $c$ such that for all $t>0$,
$$\Pr_{x \sim \on^n}\left[ |p(x)| \geq t \|p\|_2 \right] \leq 2 \cdot \exp(-c \ell \cdot t^{2/\ell}),$$ where $\|p\|_2$ is defined as $(\E_{x \sim \on^n}[ p(x)^2])^{1/2}$.
\end{lem}
In addition we will use the following simple way to convert strong concentration to a bound on expectation over subsets of assignments.
\begin{lem}
\label{lem:conc-to-exp}
Let $p(x)$ be a degree $\ell \geq 1$ polynomial over $\on^n$, let $\cS \subseteq \on^n$ be a set of assignments for which $d = 2^n/|\cS| \geq e^\ell$. Then
 $\E_{\sigma \sim \cS}[|p(\sigma)|] \leq 2 (\ln d/(c\ell))^{\ell/2} \cdot \|p\|_2$, where $c$ is the constant from Lemma \ref{lem:polynomial-concentration}.
\end{lem}
\begin{proof}
Let $c_0 = \ell \cdot c$. By Lemma \ref{lem:polynomial-concentration} we have that for any $t > 0$,
$$\Pr_{x \sim \on^n}\left[ |p(x)| \geq t \|p\|_2 \right] \leq 2 \cdot \exp(-c_0 \cdot t^{2/\ell}) .$$
The set $\cS$ contains $1/d$ fraction of points in $\on^n$ and therefore
$$\Pr_{x \sim \cS}\left[ |p(x)| \geq t \|p\|_2 \right] \leq 2 \cdot d \cdot \exp(-c_0 \cdot t^{2/\ell}) .$$

For any random variable $Y$ and value $a \in \R$, $$\E[Y] \leq a   +  \int_a^\infty \Pr[Y \geq t] dt.$$ Therefore, for $Y = |p(\sigma)|/ \|p\|_2$ and $a = (\ln d/c_0)^{\ell/2}$ we obtain \alequn{\frac{\E_{\sigma \sim \cS}[|p(\sigma)|]}{\|p\|_2}  &\leq (\ln d/c_0)^{\ell/2} + \int_{(\ln d/c_0)^{\ell/2}}^\infty d \cdot e^{-c_0 t^{2/\ell}} dt =
(\ln d/c_0)^{\ell/2} +  \frac{\ell \cdot d}{2 \cdot c_0^{\ell/2}} \cdot \int_{\ln d}^\infty e^{-z} z^{\ell/2-1} dz \\ & =
(\ln d/c_0)^{\ell/2} +  \frac{\ell \cdot d}{2 \cdot c_0^{\ell/2}} \cdot \left.\left(-e^{-z} z^{\ell/2-1}\right)\right|_{\ln d}^\infty + (\ell/2-1)\int_{\ln d}^\infty e^{-z} z^{\ell/2-2} dz = \ldots \\
 & \leq  (\ln d/c_0)^{\ell/2} + \frac{\ell \cdot d}{2 \cdot c_0^{\ell/2}} \sum_{\ell'=1/2}^{\lceil\ell/2\rceil-1} \left.\left(-\frac{\lceil\ell/2\rceil!}{\ell'!}e^{-z} z^{\ell'}\right)\right|_{\ln d}^\infty \\
 &= (\ln d/c_0)^{\ell/2} + \frac{1}{2 \cdot c_0^{\ell/2}} \sum_{\ell'=0}^{\lceil\ell/2\rceil-1} \frac{\lceil\ell/2\rceil!}{\ell'!} (\ln d)^{\ell'} \leq 2 (\ln d/c_0)^{\ell/2},
}
where we used the condition $d \geq e^{\ell}$ to obtain the last inequality.
\end{proof}
We can now use the fact that $\Gamma_\ell(\sigma,g)$ is a degree-$\ell$ polynomial of $\sigma$ to prove the following lemma:
\begin{lem}
\label{lem:gamma-norm}
Let $\cS \subseteq \on^n$ be a set of assignments for which $d = 2^n/|\cS|$.
Then $$\E_{\sigma \sim \cS}[|\Gamma_\ell(\sigma,g)|] = O_\ell\left((\ln d)^{\ell/2} \cdot \|g\|_2 /\sqrt{|X_\ell|}\right),$$
where $\|g\|_2 = \sqrt{\E_{U_\ell}[g(C_\ell)^2]}$.
\end{lem}
\begin{proof}
By Lemma \ref{lem:conc-to-exp} we get that
$$\E_{\sigma \sim \cS}[|\Gamma_\ell(\sigma,g)|] \leq 2 (\ln d/(c\ell))^{\ell/2} \cdot \|\Gamma_{\ell,g}\|_2,$$
where $\Gamma_{\ell,g}(\sigma) \equiv \Gamma_\ell(\sigma,g)$.
Now, by Parseval's identity and Lemma \ref{lem:gamma} we get that
\alequn{\E_{\sigma \sim \on^n}\left[\Gamma_{\ell,g}(\sigma)^2\right] & = \sum_{A \subseteq [n]} \widehat{\Gamma_{\ell,g}}(A)^2 \\
& = \frac{1}{|X_\ell|^2} \sum_{A \subseteq [n], |A| = \ell} \left(\sum_{C_\ell \in X_\ell, I(C_\ell)=A} g(C_\ell) \cdot (-1)^{\overline\#(C_\ell)} \right)^2 \\
& \leq \frac{1}{|X_\ell|^2} \sum_{A \subseteq [n], |A| = \ell} \left|\{ C_\ell \cond  I(C_\ell)=A\}\right| \cdot \left( \sum_{C_\ell \in X_\ell, I(C_\ell)=A} g(C_\ell)^2  \right) \\
&= \frac{2^\ell \ell!}{|X_\ell|^2} \sum_{C_\ell \in X_\ell} g(C_\ell)^2 = \frac{2^\ell \ell!}{|X_\ell|} \E_{U_\ell}[g(C_\ell)^2].
}
\end{proof}
We are now ready to bound the discrimination norm.
\begin{lem}
\label{lem:bound-dc}
Let $Q$ be a clause distribution of  the distributional complexity $r= r(Q)$, let $\D' \subseteq \{Q_\sigma\}_{\sigma \in \on^n}$ be a set of distributions over clauses and $d = 2^n/|\D'|$. Then
$\dc(\D',U_k) = O_k\left((\ln d/n)^{r/2}\right)$.
\end{lem}
\begin{proof}
Let $\cS = \{\sigma \cond Q_\sigma \in \D'\}$ and let $h:X_k \rightarrow \R$ be any function such that $\E_{U_k}[h^2] = 1$.
Let $\ell$ denote $|S|$. Using Lemma \ref{lem:delta-decompose} and the definition of $r$,
$$|\Delta(\sigma,h)| = 2^k \cdot \left|\sum_{S \subseteq [k]\setminus\{0\}} \hat{Q}(S) \cdot \Gamma_\ell(\sigma,h_S)\right| \leq 2^k \cdot  \sum_{S \subseteq [k], \ell=|S| \geq r} \left|\hat{Q}(S) \right| \cdot \left| \Gamma_\ell(\sigma,h_S)\right|.$$
Hence, by Lemma \ref{lem:gamma-norm} we get that,
\equ{ \E_{\sigma \sim \cS}[|\Delta(\sigma,h)|] \leq 2^k \cdot \sum_{S \subseteq [k],\ |S| \geq r} \left|\hat{Q}(S) \right| \cdot \E_{\sigma \sim \cS}\left[\left| \Gamma_\ell(\sigma,h_S)\right]\right| = O_k\left(\sum_{S \subseteq [k],\ |S| \geq r}\frac{(\ln d)^{\ell/2} \cdot \|h_S\|_2} {\sqrt{|X_\ell|}}\right) \label{eq:bound-kappa}}
By the definition of $h_S$,
\alequn{ \|h_S\|_2^2&= \E_{U_\ell}[h_S(C_\ell)^2]  \\
&= \frac{|X_\ell|^2}{|X_k|^2}\cdot\E_{U_\ell}\left[\left(\sum_{C\in X_k,\ C_{|S} = C_{\ell}} h(C)\right)^2\right] \\
& \leq  \frac{|X_\ell|^2}{|X_k|^2}\cdot \E_{U_\ell}\left[\frac{|X_k|}{|X_\ell|} \cdot \left(\sum_{C\in X_k,\ C_{|S} = C_{\ell}} h(C)^2\right)\right] \\
&  = \E_{U_k}[h(C)^2] = \|h\|_2^2 = 1,}
where we used Cauchy-Schwartz inequality together with the fact that for any $C_\ell$, $$\left|\{C \in X_k \cond  C_{|S} = C_{\ell} \}\right| = \frac{|X_k|}{|X_\ell|}. $$
By plugging this into eq.(\ref{eq:bound-kappa}) and using the fact that $\ln d < n$ we get,
$$ \E_{\sigma \sim \cS}[|\Delta(\sigma,h)|] = O_k\left(\sum_{\ell \geq r}\frac{(\ln d)^{\ell/2}} {\sqrt{2^\ell \cdot n!/(n-\ell)!}}\right) = O_k\left(\frac{(\ln d)^{r/2}}{n^{r/2}}\right).$$
By the definition of $\dc(\D',U_k)$ we obtain the claim.
\end{proof}

We are now ready to finish the proof of our bound on $\SDN$.
\begin{proof}(of Theorem \ref{thm:main-sdn-bound})
Our reference distribution is the uniform distribution $U_k$ and the set of distributions $\D = \D_{Q} = \{Q_\sigma\}_{\sigma \in \on^n}$ is the set of distributions for all possible assignments.
Let $\D' \subseteq \D$ be a set of distributions of size $|\D|/q$ and $\cS = \{\sigma \cond Q_\sigma \in \D'\}$. Then,  by Lemma \ref{lem:bound-dc}, we get $$\dc(\D',U_k) = O_k\left(\frac{(\ln q)^{r/2}}{n^{r/2}}\right).$$
By the definition of $\SDN$, this implies the claim.
\end{proof}

\section{Planted $k$-CSPs}
\label{sec:gen-planted}
While the focus of our presentation is on  planted satisfiability problems, the techniques can be applied to other models of planted constraint satisfaction. Here we describe describe how to apply our techniques to prove essentially identical lower bounds for the planted $k$-CSP problem.
Recall that our generalization of this planted $k$-CSP problem is defined by a function $P:\on^k \rightarrow [-1,1]$ and
we are given access to samples from a distribution $P_\sigma$, where $\sigma$ is a planted assignment in $\on^n$. A random sample from this distribution is a randomly and uniformly chosen ordered $k$-tuple of variables (without repetition) $x_{i_1},\ldots,x_{i_k}$ together with a randomly and independently chosen value $b\in \{-1,1\}$ such that $\E[b] = P(\sigma_{i_1},\ldots,\sigma_{i_k})$ (or $\Pr[b=1] = (1+ P(\sigma_{i_1},\ldots,\sigma_{i_k}))/2$). This captures the important special case when $P$ is a Boolean predicate.

Before going into the proof of the lower bound for this model we show two additional connections between this model and our planted satisfiability model. First we show that planted satisfiability can be easily reduced to the planted $k$-CSP above while preserving the complexity parameter (we remark that the reduction will always produce a non-boolean $P$ and hence requires our generalization). The second connection is that both of these models can be seen as special cases of a more general model of planted constraint satisfaction introduced by Abbe and Montanari \cite{AbbeM15}.

To describe the first reduction we start with some notation. Let $Y_k$ denote the set of all $k$-tuples of variables without repetition and let $X'_k = Y_k \times \{-1,1\}$. For a function $P:\on^k\rightarrow \R$ we use $r(P)$ to denote the degree of the lowest-degree non-zero Fourier coefficient of $P$ and refer to it as the complexity of $P$. For a clause $C=(l_1,\ldots,l_k)\in X_k$ we denote by $v(C)$ the $k$-tuple of variables in $C$ (in the same order).
For $j\in[k]$ let $s_j$ be the sign of literal $l_j$ (with $1$ meaning not negated and $-1$ meaning negated) and let $s(C) = s_1,\ldots,s_k$. We use ${\bf 1}_k$ to denote the $k$-dimensional vector $(1,1,\ldots,1)$.
\begin{lem}
\label{lem:reduce-ksat2kcsp}
There exists an algorithm that for every distribution $Q$ over $\on^k$ of complexity $r$, and any $\sigma \in \on^n$, given a random sample distributed according to $Q_\sigma$ outputs a random sample distributed according to $P_\sigma$, where $P \equiv Q - 2^{-k}$. Further, $r(Q) = r(P)$.
\end{lem}
\begin{proof}
Given a random clause $C$ the algorithm outputs the tuple of variables $v(C)$ together with a bit $b$ chosen according to the following rule. With probability $1/2$: if $s(C) = {\bf 1}_k$ then output 1, otherwise $-1$; with probability $1/2-2^{-k-1}$ output 1 and $-1$ with probability $2^{-k-1}$.

Let us analyze the resulting distribution. First we note that the output distribution is uniform over $Y_k$. This follows from the fact that for every $u\in Y_k$, $\sum_{v(C) =u} Q_\sigma(C) = \sum_{y\in\on^k} Q(y) = 1$. We now evaluate the expectation of the bit $b$ produced by our reduction as a function of $\sigma(u)$ (the values assigned by $\sigma$ to variables in $u$).  From the definition of $Q_\sigma$, for every $u\in Y_k$ and $z\in \on^k$,
\equ{\Pr_{C\sim Q_\sigma}[s(C) = z \cond v(C)=u] = Q(\sigma(C)) = Q(\sigma(u) \circ z), \label{eq:sign-distr}} where we use $\circ$ to denote the element-wise product of two vectors. In particular, $\Pr_{C\sim Q_\sigma}[s(C) = {\bf 1}_k \cond v(C)=u] = Q(\sigma(u))$. This means that $$\E[b] = \frac{1}{2}(Q(\sigma(u)) - (1-Q(\sigma(u)) + \left(\frac{1}{2} - 2^{-k-1}\right) - 2^{-k-1} = Q(\sigma(u)) - 2^{-k}.$$ This means that the reduction produces a random sample from $P_\sigma$ for $P(y) \equiv Q(y) - 2^{-k}$. Note that $\hat{P}(\emptyset) = 0$ and hence this reduction satisfies  $r(Q) = r(P)$.

\hide{
The function $w$ is implicit in the definition of the reduction: namely

\[
w(C) = \begin{cases}
1 \text{ with probability }  \forall j\in[k], s_j =1    \\
-\sgn(x_j) \text{ otherwise.}
\end{cases}
\]
}
\end{proof}

We now show how both of these models can be seen as special cases of the model in \cite{AbbeM15}. The model is specified by a collection of distributions $\{\Phi(\cdot \cond y)\}_{y\in\on^k}$ over some output alphabet $Z$. For a planted assignment $\sigma\in\on^n$ (their model allows a more general alphabet for each variable but $\on$ suffices to subsume the models discussed in this paper) the planted distribution $\Phi_\sigma$ is defined as follows. A random sample from this distribution is a randomly and uniformly chosen ordered $k$-tuple of variables $u \in Y_k$ together with value $z$ chosen randomly and independently according to $\Phi(\cdot \cond \sigma(u))$. We first observe that for any $P:\on^k \rightarrow[-1,1]$, setting $Z=\on$ and having $\Phi(b \cond y) = (1 + b\cdot P(y))/2$ recovers exactly the generalized Goldreich's planted $k$-CSP for function $P$.

To recover the planted satisfiability model for distribution $Q$, we let $Z=\on^k$ and then define $\Phi(z \cond y) = Q(y \circ z)$. Here the output alphabet represents the negation signs of variables. A $k$-tuple of variables $u\in Y_k$ with $k$ negation signs $z\in \on^k$ uniquely describes a clause $C\in X_k$ such that $v(C) =u$ and $s(C)=z$. Further, by Eqn.~\eqref{eq:sign-distr}, we get that $\Phi_\sigma$ for $\Phi$ defined as above is exactly $Q_\sigma$. It is not hard to see that the techniques in this work can also be applied to characterize the SQ complexity of solving planted $k$-CSPs in this more general model.

\subsection{Lower Bounds for Planted $k$-CSPs}
\label{sec:csp-extend}
We prove the analogue of Theorem \ref{thm:main-sdn-bound} for the planted $k$-CPS, which, in turn, immediately implies that the lower bounds stated in Theorem \ref{thm:lower-bound} apply to this problem verbatim. We first note that the reduction in Lemma \ref{lem:reduce-ksat2kcsp} implies the desired lower bound for all functions $P$ such that $P \equiv Q -2^{-k}$ for some distribution $Q$ over $\on^k$. Unfortunately, this is not sufficient to obtain a lower bound for all functions $P:\on^{-k} \rightarrow [-1,1]$. Indeed, this does not give a lower bound for any Boolean $P$. At the same time, we show that the reduction in Lemma \ref{lem:reduce-ksat2kcsp} can be used to reduce bounds on the discrimination norm of the planted $k$-CSP problem to the bounds on the discrimination norm for planted satisfiability that we gave in Section \ref{sec:plantstat-sdn}\footnote{A direct proof of this bound can be found in an earlier version of this work \cite[v5]{FeldmanPV:13}.}. We are not aware of similar reductions in the literature and our technique might be useful for relating the complexity of other problems for which standard reductions are not known.

We now give the formal details. Let $P:\on^k\rightarrow [-1,1]$ be a function on $k$-bits. Let $\D_P$ denote the set of all distributions $P_\sigma$, where $\sigma \in \on^n$ and $U'_k$ be the uniform distribution over $X'_k = Y_k \times \{-1,1\}$. Let $\B(\D_P,U'_k)$ denote the decision problem in which given samples from an unknown input distribution $D \in \D_P \cup \{U'_k\}$ the goal is to output $1$ if $D\in \D_P$ and 0 if $D=U'_k$. Our goal is to prove the following results.
\begin{thm}
\label{thm:sdn-bound-csp}
 For any function $P:\on^k\rightarrow [-1,1]$ of complexity $r$, there exist a constant $c > 0$ (that depends on $P$) such that for any $q \geq 1$, $$\SDN\left(\B(\D_P,U'_k), \frac{c (\log q)^{r/2}}{n^{r/2}}\right) \geq q.$$
\end{thm}
As in the case of Theorem \ref{thm:main-sdn-bound}, it suffices to prove the following analogue of  Lemma \ref{lem:bound-dc}.
\begin{lem}
\label{lem:bound-dc-csp}
  Let $P:\on^k\rightarrow [-1,1]$ be any function of complexity $r= r(P)$, let $\D' \subseteq \{P_\sigma\}_{\sigma \in \on^n}$ be a set of distributions over clauses and $d = 2^n/|\D'|$. Then
$\dc(\D',U'_k) = O_k\left((\ln d/n)^{r/2}\right)$.
\end{lem}
\begin{proof}
We first note that this bound does not say anything non-trivial for $r =0$ (and, indeed, the label distribution is biased in this case and can be distinguished from $U'_k$ using a constant number of samples). Therefore, from now on we assume that $r\geq 1$. Let $\cS = \{\sigma \cond P_\sigma \in \D'\}$ and let $h':X'_k \rightarrow \R$ be any function such that $\|h'\|_{U'_k} = 1$ and $\E_{\sigma \sim \cS}\left[\left| \E_{P_\sigma}[h']-\E_{U'_k}[h']\right| \right] = \dc(\D',U'_k)$.
We define a function $h$ on $X_k$ as follows. If for $C\in X_k$, $s(C) = {\bf 1}_k$ then $h(C) =h'(v(C),1)$, otherwise $h(C) =h'(v(C),-1)$. We now claim that for every $\sigma \in \on^n$,
\equ{\E_{P_\sigma}[h'] - \E_{U'_k}[h'] = 2^{2k-1} \cdot \left(\E_{Q_\sigma}[h] - \E_{U_k}[h]\right), \label{eq:reduce-delta}}
where $Q \equiv (P+1)/2^k$. Note that $Q$ defined in this way is a distribution since for all $y\in \on^k$, $Q(y) \geq 0$ and $\sum_{y\in \on^k} Q(y) = 2^k \cdot \hat{P}(\emptyset) +1 = 1$.

Distributions $Q_\sigma$,$P_\sigma$ $U_k$ and $U'_k$ are uniform over $k$-tuples of variables and therefore to prove eq.~\eqref{eq:reduce-delta}, it suffices to prove that for every $u \in Y_k$,
\alequ{\E_{(v,b)\sim P_\sigma}[h'(v,b) \cond v=u ] &- \E_{(v,b)\sim U'_k}[h'(v,b) \cond v=u ] \nonumber\\
&= 2^{2k-1} \cdot\left(\E_{C \sim Q_\sigma}[h(C) \cond v(C) = u] - \E_{C \sim U_k}[h(C) \cond v(C) = u] \right). \label{eq:reduce-delta-u}}
The left hand side of this equality is equal to $$h'(u,1) \cdot \frac{1+P(\sigma(u))}{2} + h'(u,-1) \cdot \frac{1-P(\sigma(u))}{2}  - \frac{1}{2} \cdot h'(u,1) +\frac{1}{2} \cdot h'(u,-1) = \frac{P(\sigma(u))\cdot (h'(u,1) - h'(u,-1))}{2} .$$
By equation \eqref{eq:sign-distr}, the right side of eq.~\ref{eq:reduce-delta-u} is equal to
\alequn{2^{2k-1} \cdot 2^{-k} \cdot &\sum_{C, v(C)=u} h(C) \cdot (Q(C)-1)) =\\
& = 2^{k-1} \cdot\left( h'(u,1) \cdot \frac{P(\sigma(u))}{2^k} +  \sum_{C, v(C)=u, s(C)\neq {\bf 1}_k} h'(u,-1) \cdot \frac{P(\sigma(u) \circ s(C))}{2^k} \right) \\ &= \frac{1}{2} \cdot\left( h'(u,1) \cdot P(\sigma(u)) -  h'(u,-1) \cdot P(\sigma(u)) \right) = \frac{P(\sigma(u))\cdot (h'(u,1) - h'(u,-1))}{2},
}
where we used the fact that $\sum_{y\in\on^k}P(y) = 2^k \cdot \hat{P}(\emptyset)=0$ to obtain the equality of the second line to the third one.

Now all we need to bound $\dc(\D',U'_k)$ is an upper bound on $\|h\|_{U_k}$.
First, note that by our assumption,
\equ{\E_{U'_k}[h'^2] = \frac{1}{|Y_k|} \cdot \sum_{u\in Y_k}\frac{1}{2} \left(h'(u,1)^2 + h'(u,-1)^2\right) = 1. \label{eq:old-norm-bound}}

For every $v \in Y_k$, $$\E_{C \sim U_k}[(h(C))^2 \cond v(C) = u] =  \frac{1}{2^k} \cdot \left(h'(u,1)^2 +  (2^k-1) h'(u,-1)^2\right) \geq \frac{1}{2^k}\left(h'(u,1)^2 +  h'(u,-1)^2\right).$$ Using eq.~\eqref{eq:old-norm-bound} we get that,
$$\|h\|_{U_k}^2 \geq \frac{1}{|Y_k|} \cdot \frac{1}{2^k} \cdot \sum_{u\in Y_k}  \left(h'(u,1)^2 +  h'(u,-1)^2\right) = \frac{1}{2^{k-1}}.$$

Using this bound on the norm and eq.~\eqref{eq:reduce-delta} we can now bound  $\dc(\D',U'_k)$ as follows. Let $\D'_Q \doteq \{Q_\sigma \cond \sigma \in \cS\}$.

\alequn{\dc(\D',U'_k) &= \E_{\sigma \sim \cS}\left[\left| \E_{P_\sigma}[h']-\E_{U'_k}[h']\right| \right] =  2^{2k-1} \cdot \E_{\sigma \sim \cS}\left[\left|\E_{Q_\sigma}[h] - \E_{U_k}[h]\right|\right] \\
&\leq2^{2k-1} \cdot \frac{\dc(\D'_Q,U_k)}{\|h\|_{U_k}} \leq 2^{2k-1 + k/2-1/2} \cdot \dc(\D'_Q,U_k) =  O_k\left((\ln d/n)^{r/2}\right),
}
where we used Lemma \ref{lem:bound-dc} to obtain the last bound.
\end{proof}

\section{Lower Bounds using Statistical Dimension}
\label{sec:general-stat-bounds}
\subsection{Lower bound for $\VSTAT$}
We first prove an analogue of lower-bound for $\VSTAT$ from \cite{FeldmanGRVX:12} but using the statistical dimension based on discrimination norm instead of the average correlation. It is not hard to see that discrimination norm is upper-bounded by the square root of average correlation and therefore our result subsumes the one in \cite{FeldmanGRVX:12}. 
\begin{thm}
\label{thm:avgvstat-random-app-decision}
Let $X$ be a domain and $\B(\D,D)$ be a decision problem over a class of distributions $\D$ on $X$ and reference distribution $D$.
Let $d = \SDN(\B(\D,D),\kappa)$ and let $\D_D$ be a set of distributions for which the value $d$ is attained. Consider the following average-case version of the $\B(\D,D)$ problem: the input distribution $D'$ equals $D$ with probability $1/2$ and $D'$ equals a random uniform element of $\D_D$ with probability $1/2$. Any randomized statistical algorithm that solves $\B(\D,D)$ with success probability $\gamma > 1/2$ over the choice of $D'$ and randomness of the algorithm requires at least $(2\gamma - 1) d$ calls to $\VSTAT(1/(3\kappa^2))$.
\end{thm}
\begin{proof}
  We prove our lower bound for any deterministic statistical algorithm and the claim for randomized algorithms follows from the fact that the success probability of a randomized algorithm is just the expectation of its success probability for a random fixing of its  coins.

  Let $\A$ be a deterministic statistical algorithm that uses $q$ queries to $\VSTAT(1/(3\kappa^2))$ to solve $\B(\D,D)$ with probability $\gamma$ over a random choice of an input distribution described in the statement. Following an approach from \cite{Feldman:12jcss}, we simulate $\A$ by answering any query $h:X \rightarrow [0,1]$ of $\A$ with value $\E_D[h(x)]$. Let $h_1,h_2,\ldots,h_q$
  be the queries asked by $\A$ in this simulation and let $b$ be the
  output of $\A$. $\A$ is successful with probability $\gamma > 1/2$ and therefore $b=0$, that is $\A$ will certainly decide that the input distribution equals to $D$.

  Let the set $\D^+ \subseteq \D_D$ be the set of
  distributions on which $\A$ is successful (that is outputs $b=1$) and we denote these distributions
  by $\{D_1,D_2,\ldots,D_m\}$. We recall that, crucially, for $\A$ to be considered successful it needs to be successful for any valid responses of $\VSTAT$ to $\A$'s queries. We note that the success probability of $\A$ is $\frac{1}{2}+ \frac{1}{2} \frac{m}{|\D_D|}$ and therefore $m \geq (2\gamma -1)|\D_D|$.

  For every $k \leq q$, let  $A_k$ be the set of all distributions $D_i$ such
  that $$\left|\E_{D}[h_k(x)] - \E_{D_i}[h_k(x)]\right| > \tau_{i,k}
  \doteq \max\left\{\frac{1}{t},
    \sqrt{\frac{p_{i,k}(1-p_{i,k})}{t}}\right\},$$ where we use $t$ to
  denote $1/(3\kappa^2)$ and $p_{i,k}$ to denote
  $\E_{D_i}[h_k(x)].$ To prove the desired bound we first prove the
  following two claims:
  \begin{enumerate}
  \item $\sum_{k\leq q}|A_k| \geq m$;
  \item for every $k$, $|A_k| \leq |\D_D|/d$.
  \end{enumerate}
  Combining these two implies that $q \geq d \cdot m/|\D_D|$ and therefore $q \geq (2\gamma -1) d$ giving the desired lower bound.

  In the rest of the proof for conciseness we drop the subscript $D$
  from inner products and norms. To prove the first claim we assume, for the sake of contradiction, that there exists $D_i \not\in \cup_{k\leq q} A_k$. Then for every $k\leq q$, $|\E_{D}[h_k(x)] - \E_{D_i}[h_k(x)]| \leq \tau_{i,k}$. This implies that the replies of our simulation $\E_{D}[h_k(x)]$ are within $\tau_{i,k}$ of $\E_{D_i}[h_k(x)]$, in other words are valid responses. However we know that for these responses $\A$ outputs $b=0$ contradicting the condition that $D_i \in \D^+$.

To prove the second claim, suppose that for some $k \in [d]$, $|A_k| > |\D_D|/d$. Let $p_k = \E_{D}[h_k(x)]$ and assume that $p_k \leq 1/2$ (when $p_k>1/2$ we just replace $h_k$ by $1-h_k$ in the analysis below). We will next show upper and
  lower bounds on the following quantity
\equ{\Phi = \sum_{D_i \in A_k}\left[\left|\E_{D}[h_k(x)] - \E_{D_i}[h_k(x)]\right|\right] = \sum_{D_i \in A_k}|p_k - p_{i,k}|.\label{eq:define-phi}}

By our assumption for $D_i \in A_k$, $|p_{i,k}-p_k| > \tau_{i,k} = \max\{1/t, \sqrt{p_{i,k}(1-p_{i,k})/t}\}$. If $p_{i,k} \geq 2p_k/3$ then $$|p_k-p_{i,k}| > \sqrt{\frac{p_{i,k}(1-p_{i,k})}{t}} \geq \sqrt{\frac{\frac{2}{3}p_k \cdot \frac{1}{2}}{t}} = \sqrt{\frac{p_k}{3t}}.$$ Otherwise (when $p_{i,k} < 2p_k/3$),  $p_k - p_{i,k} > p_k - 2p_k/3 = p_k/3$.
We also know that $|p_{i,k}-p_k| > \tau_{i,k} \geq 1/t$ and therefore $|p_{i,k}-p_k| > \sqrt{\frac{p_k}{3t}}$.
Substituting this into eq.~(\ref{eq:define-phi}) we get that \equ{\Phi > |A_k| \cdot \sqrt{\frac{p_k}{3t}} = |A_k| \cdot \sqrt{p_k} \cdot \kappa. \label{eq:low-bound-phi}}

Now, by the definition of discrimination norm and its linearity we have that
$$\sum_{D_i \in A_k}\left[\left|\E_{D}[h_k(x)] - \E_{D_i}[h_k(x)]\right|\right] = |A_k| \cdot \E_{D' \sim A_k}\left[\left|\E_{D}[h_k(x)] - \E_{D'}[h_k(x)]\right|\right] \leq |A_k| \cdot \dc(A_k,D) \cdot \|h_k\|_2 .$$
We note that, $h_k$ is a $[0,1]$-valued function and therefore $\|h_k\|^2 = \E_D[h_k(x)^2] \leq \E_D[h_k(x)] = p_k$. Also by definition of $\SDN$, $\dc(A_k,D) \leq \kappa$. Therefore $\Phi \leq |A_k| \cdot \kappa \cdot \sqrt{p_k}$.
This contradicts the bound on $\Phi$ in eq.~(\ref{eq:low-bound-phi}) and hence finishes the proof of our claim.
\end{proof}

\subsection{Lower bounds for $\MVSTAT$ and $\MSAMPLE$}
\newcommand{\zL}{L_0}
We now describe the extension of our lower bound to $\MVSTAT$ and $\MSAMPLE(L)$ oracles. For simplicity we state them for the worst case search problems but all these results are based on a direct simulation of an oracle using a $\VSTAT$ oracle and therefore they equivalently apply to the average-case versions of the problem defined in Theorem \ref{thm:avgvstat-random-app-decision}.

Given the lower bound $\VSTAT$ we can obtain our lower bound for $\MVSTAT$ via the following simple simulation. For conciseness we use $\zL$ to denote $\{0,1,\ldots,L-1\}$.
\begin{thm}
\label{thm:simulate-mvstat}
Let $D$ be the input distribution over the domain $X$, $t , L > 0$ be integers.
For any multi-valued function $h:X \rightarrow \zL$ and any set $\cS$ of subsets of $ \zL$, $L$ queries to $\VSTAT(4L\cdot t)$ can be used to give a valid answer to query $h$ with set $\cS$ to $\MVSTAT(L,t)$.
\end{thm}
\begin{proof}
For $i \in \zL$ we define $h_i(x)$ as $h_i(x) = 1$ if $h(x) = i$ and $0$ otherwise.
Let $v_i$ be the response of $\VSTAT(4L\cdot t)$ on query $h_i$.
For any $Z \subseteq \zL$,
\alequn{\left|\sum_{\ell \in Z} v_\ell - p_Z \right| & \le \sum_{\ell \in Z} |v_i - p_i| \\
& \leq \sum_{\ell \in Z} \max \left\{\frac{1}{4Lt}, \sqrt{\frac{p_i(1-p_i)}{4Lt}}\right\} \\
& \leq \frac{|Z|}{4Lt} + \sum_{\ell \in Z} \sqrt{\frac{p_i(1-p_i)}{4Lt}} \\
& \leq \frac{|Z|}{4Lt} + \sqrt{|Z|} \cdot \sqrt{\frac{\sum_{\ell \in Z} p_i(1-p_i)}{4Lt}}\\
& \leq \frac{|Z|}{4Lt} + \sqrt{|Z|} \cdot \sqrt{\frac{ p_Z(1-p_Z)}{4Lt}}\\
& \leq \frac{1}{4t} + \sqrt{\frac{ p_Z(1-p_Z)}{4t}}\\
& \leq \max \left\{\frac{1}{t}, \sqrt{\frac{p_Z(1-p_Z)}{t}}\right\} ,}
where $p_Z = \Pr_D[h(x) \in Z]$.
\end{proof}

We now describe our lower bound for $\MSAMPLE(L)$ oracle.
\begin{thm}
\label{thm:avgsample-v}
  Let $\B(\D,D)$ be a decision problem. For $\kappa > 0$, let $d = \SDN(\B(\D,D),\kappa)$.
  Any (possibly randomized) statistical algorithm that solves $\B(\D,D)$ with probability $\gamma>1/2$ requires at least $m$ calls to $\MSAMPLE(L)$ for $$m  = \Omega\left(\frac{1}{L} \min\left\{d(2\gamma -1),\frac{\gamma^2}{\kappa^2} \right\}\right)\ .$$ In particular, any algorithm with success probability of at least $2/3$ requires at least $\Omega\left(\frac{1}{L} \cdot \min\{d,1/\kappa^2\}\right)$ samples from $\MSAMPLE(L)$.
\end{thm}
The proof of this result is based on the following simulation of $\MSAMPLE(L)$ using $\VSTAT$.
\begin{thm}
\label{th:unbiased-from-vstat}
Let $\Z$ be a search problem and let $\A$ be a (possibly randomized) statistical algorithm that solves $\Z$ with probability at least $\gamma$ using $m$ samples from $\MSAMPLE(L)$.  For any $\delta \in (0,1/2]$, there exists a statistical algorithm $\A'$ that uses at most $O(m \cdot L)$ queries to $\VSTAT(L \cdot m/\delta^2)$ and solves $\Z$ with probability at least $\gamma - \delta$.
\end{thm}
A special case of this theorem for $L=2$ is proved in \cite{FeldmanGRVX:12}. Their result is easy to generalize to the statement of Theorem \ref{th:unbiased-from-vstat} but is it fairly technical. Instead we describe a simple way to simulate $m$ samples of $\MSAMPLE(L)$ using $O(mL)$ samples from $\SAMPLE$. This simulation (together with the simulation of $\SAMPLE$ from \cite{FeldmanGRVX:12}) imply Theorem \ref{th:unbiased-from-vstat}. It also allows to easily relate the powers of these oracles. The simulation is based on the following lemma (proof by Jan Vondrak).
\begin{lem}
\label{lem:simulate-1-mstat}
Let $D$ be the input distribution over $X$ and let $h:X \rightarrow \zL$ be any function. Then using $L+1$ samples from $\SAMPLE$ it is possible to output a random variable $Y \in \zL \cup \{\bot\}$, such that
\begin{itemize}
\item $\Pr[Y \neq \bot] \geq 1/(2e)$,
\item for every $i \in \zL$, $\Pr[Y=i \cond Y\neq \bot] = p_i$.
\end{itemize}
\end{lem}
\begin{proof}
$Y$ is defined as follows. For every $i \in \zL$ ask a sample for $h_i$ from $\SAMPLE$ and let $B_i$ be equal to the outcome with probability $1/2$ and $0$ with probability $1/2$ (independently). If the number of $B_i$'s that are equal to $1$ is different from 1 then $Y=\bot$. Otherwise let $j$ be the index such that $B_j = 1$. Ask a sample for $h_j$ from $\SAMPLE$ and let $B'_j$ be the outcome with probability $1/2$ and 0 with probability $1/2$. If $B'_j =0$ let $Y=j$, otherwise $Y=\bot$.
From the definition of $Y$, we obtain that for every $i \in \zL$,
$$\Pr[Y=i]=\frac{p_i}{2} \cdot \prod_{k\neq i} (1-\frac{p_k}{2}) \cdot (1-\frac{p_i}{2}) = \frac{p_i}{2} \cdot \prod_{k \in \zL} (1-\frac{p_k}{2}). $$
This implies that for every $i \in \zL$, $\Pr[Y=i \cond Y\neq \bot] = p_i$. Also
$$\Pr[Y \neq \bot] = \sum_{i \in \zL}\frac{p_i}{2} \cdot \prod_{i \in \zL} (1-\frac{p_i}{2}) \geq \frac{1}{2}\prod_{k \in \zL} e^{-p_i} =e^{-1}/2, $$
where we used that for $a\in [0,1/2]$, $(1-a)\leq e^{-2a}$.
\end{proof}
Given this lemma we can simulate $\MSAMPLE(L)$ by sampling $Y$ until $Y \neq \bot$. It is easy to see that simulating $m$ samples from $\MSAMPLE(L)$ will require at most $4e \cdot m(L+1)$ with probability at least $1-\delta$ for $\delta$ exponentially small in $m$.

We now combine Theorems \ref{thm:avgvstat-random-app-decision} and \ref{th:unbiased-from-vstat} to obtain the claimed lower bound for  statistical algorithms using $\MVSTAT$.
\begin{proof}[Proof of Theorem \ref{thm:avgsample-v}]
Assuming the existence of a statistical algorithm using less than $m$ samples we apply Theorem \ref{th:unbiased-from-vstat} for $\delta = \gamma/2 - 1/4$ to simulate the algorithm using $\VSTAT$. The bound on $m$ ensures that the resulting algorithm uses less than $\Omega{\left(d(2\gamma-1)\right)}$ queries to $\VSTAT(\frac{1}{3\kappa^2})$ and has success probability of at least $\gamma/2 + 1/4$. By substituting these parameters into Theorem \ref{thm:avgvstat-random-app-decision} we obtain a contradiction.
\end{proof}

Finally we state an immediate corollary of Theorems \ref{thm:avgvstat-random-app-decision}, \ref{thm:simulate-mvstat} and \ref{thm:avgsample-v} that applies to general search problems and generalizes Theorem \ref{thm:avgvstat-random}.
\begin{thm}
\label{thm:avgvstat-random-general}
  Let $\B(\D,D)$ be a decision problem. For $\kappa > 0$, let $d = \SDN(\B(\D,D),\kappa)$ and let $L \geq 2$ be an integer. Any randomized statistical algorithm that solves $\B(\D,D)$ with probability $\geq2/3$ requires either
  \begin{itemize}
  \item   $\Omega(d/L)$ calls to $\MVSTAT(L,1/(12 \cdot \kappa^2 \cdot L))$;
  \item   at least $m$ calls to $\MSAMPLE(L)$ for $m = \Omega\left(\min\left\{d,1/\kappa^2 \right\}/L\right)$.
  \end{itemize}
\end{thm}


\section{Algorithmic Bounds}
\label{sec:algsec}
In this section we prove Theorem \ref{thm:algo1}.  The algorithm is a variant of the subsampled power iteration from \cite{FeldmanPV14} that can be implemented statistically. We describe the algorithm for the planted satisfiability model, but it can be adapted to solve Goldreich's planted $k$-CSP by considering only the $k$-tuples of variables that the predicate $P$ evaluates to $1$ on the planted assignment $\sigma$.   

\subsection{Set-up}

Lemma 1 from \cite{FeldmanPV14} states that subsampling $r$ literals from a distribution  $Q_\sigma$ on $k$-clauses with distribution complexity $r$ and planted assignment $\sigma$ induces a parity distribution over clauses of length $r$, that is a distribution over $r$-clauses with planting function $Q^\delta: \{\pm 1\}^r \to \R^+$ of the form
$Q^\delta(x) = \del/2^{r}$ for $|x|$ even, $Q^\delta(x) = (2-\del)/2^{r}$ for $|x|$ odd,
for some $\del \in [0,2]$ , $\del \ne 1$, where $|x|$ is the number of $+1$'s in the vector $x$.  The set of $r$ literals to subsample from each clause is given by the set $S \subset \{ 1, \ldots, k \}$ with $\hat Q(S) \ne 0$.

From here on the distribution on clauses will be given by $Q_\sigma^{\delta}$, for $\delta \neq 1$ and planted assignment $\sigma$. For ease of analysis, we define $Q_{\sigma,p}$  as the distribution over $k$-clause formulas in which each possible $k$-clause with an even number of true literals under $\sigma$ appears independently in $Q_{\sigma,p}$ with probability $\del p$, and each clause with and odd number of true literals appears independently with probability $(2- \del )p$, for an overall clause density $p$.  We will be concerned with $p = \tilde \Theta(n^{-k/2})$. Note that it suffices to solve the algorithmic problem for this distribution instead of that of selecting exactly $m = \tilde \Theta (n^{k/2})$ clauses independently at random.  In particular, with probability $1- \exp( - \Theta(n))$, a sample from $Q_{\sigma, p}$ will contain at most $2p \cdot  \frac{2^k n!}{(n-k)!} = O(n^k p)$ clauses.

We present statistical algorithms to recover the partition of the $n$ variables into positive and negative literals.  We will recover the partition, which gives $\sigma$ up to a sign change.

The algorithm proceeds by constructing a  biadjacency matrix $M$ of size $N_1 \times N_2$ with
$N_1 = 2^{\lceil k/2 \rceil} \frac{n!}{(n - \lceil k/2 \rceil)!}$, $N_2 = 2^{\lfloor k/2 \rfloor} \frac{n!}{(n-\lfloor k/2 \rfloor)!}$. We set $N = \sqrt{N_1 N_2}$. For even $k$, we have $N_1 = N_2 = N$ and thus $M$ is a square matrix.  The rows of the matrix are indexed by ordered subsets $S_1, \ldots, S_{N_1}$ of $\lceil k/2 \rceil$ literals and columns by subsets $T_1, \ldots, T_{N_2}$ of $\lfloor k/2 \rfloor$ literals.  For a formula $\mathcal F$, we construct a matrix $\hat M(\mathcal F)$ as follows.  For each $k$-clause $(l_1, l_2, \dots, l_k)$ in $\mathcal F$, we put a $1$ in the entry of $\hat M$ whose row is indexed by the set $(l_1, \dots, l_{\lceil k/2 \rceil}) $ and column by the set $( l_{\lceil k/2 \rceil +1}, \dots, l_{k})$.

Let $M_{\sigma,p}$ denote the distribution  on random $N_1 \times N_2$ matrices induced by drawing a random formula according to $Q_{\sigma,p}$ and forming the associated matrix $ M(Q_{\sigma,p})$ as above.

For $k$ even, let $u \in \{ \pm 1 \}^N$ be the vector with a $+1$ entry in every coordinate indexed by subsets containing an even number of true literals under $\sigma$, and a $-1$ entry for every odd subset.  For $k$ odd,  define the analogous vectors $u_y \in  \{ \pm 1\}^{N_1}$ and $u_x \in \{ \pm 1 \}^{N_2}$, again with $+1$'s for even subsets and $-1$ for odd subsets.

The algorithm will apply a modified power iteration procedure with rounding to find $u$ or $u_x$ (up to a change of sign).  From these vectors the partition $\sigma$ into true and false literals can be determined by solving a system of linear equations.

For even $k$, the discrete power iteration begins by sampling a random vector $x^0 \in \{ \pm 1 \}^N$ and multiplying by a sample of $M_{\sigma,p}$.  We then randomly round each coordinate of $M_{\sigma,p} x^0$ to $\pm 1$ to get $x^1$, and then repeat, drawing a fresh sample from $M_{\sigma,p}$ at each step.  The rounding at each is probabilistic and depends on the value of each coordinate and the maximum value of all the coordinates.  The number of clauses used by the algorithm is the sum of the number of clauses used in each sampled matrix, or in other words, if we use $T$ samples from $M_{\sigma, p}$, our original formula needs density $T \cdot p$. \footnote{Obtaining $T$ (nearly) independent samples of $M_{\sigma,p}$ from one sample of $M_{\sigma, Tp}$ is a subtle issue and is addressed in \cite{FeldmanPV14}; for the purposes of the implementation by statistical oracles this is irrelevant and so we avoid the discussion here.}

For odd $k$, we begin with a random $x^ 0 \in \{\pm 1 \}^{N_2}$ and a random sample $M_{\sigma,p}$, then form $y^0$ by deterministically rounding $M_{\sigma,p} x^0$ to a vector with entries $-1, 0,$ or $+1$.  Then we form $x^1$ by taking a fresh sample of $M_{\sigma,p}$ and perform a randomized $\pm 1$ rounding of $M_{\sigma,p}^T y^0$, and repeat.  
  There is a final rounding step to find a $\pm 1$ vector that matches $u$ or $u_x$.

 In Section~\ref{sec:statImplement} we will prove that this algorithm can be implemented {\em statistically} in any of the following ways:
\begin{enumerate}
\item Using $O(n^{r/2} \log^2 n)$ calls to $\MSAMPLE(n^{\lceil r/2\rceil})$;
\item For even $r$: using $O(\log n)$ calls to $\MVSTAT(n^{r/2},n^{r/2} \log \log n)$;
\item For odd $r$: using $O(\log n)$ calls to $\MVSTAT(O(n^{\lceil r/2\rceil}), O(n^{r/2}\log n))$. 
\end{enumerate}

\subsection{Algorithm Discrete-Power-Iterate (even $k$).}

\begin{enumerate}
\item  Pick $x^0 \in \{ \pm 1 \}^N$ uniformly at random.
For $i = 1, \dots \log N$, repeat the following:
\begin{enumerate}
\item Draw a sample matrix $M \sim M_{\sigma,p}$.
\item Let $x = M x^{i-1} $.
\item Randomly round each coordinate of $x$ to $\pm 1$ to get $x^i$ as follows:
let
\[
x^i_j = \begin{cases}
\sgn(x_j) \text{ with probability }   \frac{1}{2} + \frac{ |x_j|  } {2 \max_j |x_j| }   \\
-\sgn(x_j) \text{ otherwise.}
\end{cases}
\]
\end{enumerate}
\item Let $x = M x^{\log N} $ and set $u^* = \sgn(x)$ by rounding each coordinate to its sign.
\item Output the solution by solving the system of parity equations defined by $u^*$.
\end{enumerate}

\begin{lem}
\label{oddnewlem}
If $p = \frac{K \log N}{(\del-1)^2 N}$ for a sufficiently large constant $K$, then with probability $1-o(1)$ the above algorithm returns the planted assignment.
\end{lem}

The main idea of the analysis is to keep track of the random process $(u \cdot x^i)$.  It starts at $\Theta(\sqrt N)$ with the initial randomly chosen vector $x^0$, and then after an initial phase, doubles on every successive step whp until it reaches $N/9$.

We will use the following Chernoff bound several times (see eg. Corollary A.1.14 in \cite{alon2011probabilistic}).
\begin{prop}
\label{ChernoffProp}
Let $X = \sum_{i=1}^m \xi_i Y_i$ and $Y = \sum_{i=1}^m Y_i$, where the $Y_i$'s are independent Bernoulli random variables and the $\xi_i$'s are fixed $\pm 1$ constants.  Then
\[ \Pr[ |X - \E[X]| \ge \alpha \E[Y] ] \le e^{ - \alpha^2\E[Y]/3}.  \]
\end{prop}

\begin{prop}
\label{evenauxprop}
If $|x^i \cdot u| = \beta N \ge \sqrt N  \log \log N$, then with probability $1- O(1/ N\beta^2 )$,
\[ |x^{i+1} \cdot u| \ge \min \left \{ \frac{ N}{9}, 2|x^i \cdot u|  \right \}. \]
\end{prop}

\begin{proof}
We assume WLOG that $\del >1$ and $x^0\cdot u>0$ in what follows.  Let $U ^+ = \{ i : u_i = +1\}$, $U ^- = \{ i : u_i = -1\}$, $X ^+ = \{ i : x_i = +1\}$, and $X ^- = \{ i: x_i = +1\}$.  For a given $j\in [N]$, let $A_j = \{ i : \text {sets } i \text{ and } j \text{ share no variables} \}$.  We have $|A_j| = N^*$ for all $j$.

Let $z=Mx^i$. Note that the coordinates $z_1, \dots z_N$ are independent and if $j \in U^+$,
\begin{align*}
z_j \sim Z_{++} + Z_{-+} - Z_{+-} - Z_{--} -  p(|X^+| - |X^-| )
\end{align*}
where
\begin{align*}
Z_{++} &\sim Bin( |U^+ \cap X^+ \cap A_j|, \del p), \\
Z_{-+} &\sim Bin( |U^- \cap X^+ \cap A_j|, (2-\del) p), \\
Z_{+-} &\sim Bin( |U^+ \cap X^- \cap A_j|, \del p) ,\\
Z_{--} &\sim Bin( |U^- \cap X^- \cap A_j|, (2-\del) p).
\end{align*}
We can write a similar expression if $j \in U^-$, with the probabilities swapped. For $j \in U^+$ we calculate,
\begin{align*}
\E[z_j] &= \del p |U^+ \cap X^+ \cap A_j| +(2-\del)p |U^- \cap X^+ \cap A_j|  - \del p |U^+ \cap X^- \cap A_j| - (2-\del)p |U^- \cap X^- \cap A_j|   \\
& -  p \left(|X^+ | - |X^- |   \right)\\
&= \del p |U^+ \cap X^+| + (2-\del)p | U^- \cap X^+| - \del p |U^+ \cap X^-| \\
& - (2-\del)p |U^- \cap X^-|  - p \left(|X^+| - |X^-|   \right) + O((N-N^*)p)\\
&= (\del - 1)p (u \cdot x) + O( n^{k/2 -1} p).
\end{align*}
For $j \in U^-$ we get $\E[z_j] = (1- \del) p (u \cdot x) + O( n^{k/2 -1} p) $.

To apply Proposition \ref{ChernoffProp}, note that there are $N$ entries in each row of $M$, half of which  have probability $\del p$ of being $1$, and  the other half with probability $(2- \del)p$ of being a $1$, so $\E[Y] = N p$.  Using the proposition with $\alpha=(\del -1)/26$ and union bound, we have that with probability $1- o(N^{-1})$,
\begin{align}
\label{eq:maxjeven}
\max_j |z_j| &\le (\del -1) p \cdot (u \cdot x)  +\frac{(\del-1) Np}{26}  + O(n^{k/2 -1} p) \\
\nonumber
&\le (\del -1) p \cdot (u \cdot x)  +\frac{(\del-1) Np}{25}.
\end{align}
For each ordered set of $k/2$ literals indexed by $j \in U^+$, there is a set indexed by $j^\prime \in U^-$ that is identical except the first literal in $j^\prime$ is the negation of the first literal in $j$.  Note that $A_j = A_{j^\prime}$, and so we can calculate:
\begin{align*}
\E[z_j] - \E[z_{j^\prime}] &= 2 (\del - 1)p \left [ |U^+ \cap X^+ \cap A_j| + | U^-\cap X^- \cap A_j| -|U^- \cap X^+ \cap A_j| - | U^+\cap X^- \cap A_j|  \right ]
\end{align*}
which is simply  $2 (\del -1) p$ times the dot product of $u$ and $x$ restricted to the coordinates $A_j$.
  Summing over all $j \in [N]$ we get
\begin{align*}
\E[u \cdot z] &= |A_j| (\del - 1)p (u \cdot x) \\
&=   N^*(\del -1)p (u \cdot x)\\
&= N(\del -1)p (u \cdot x) (1+o(1)).
\end{align*}

Applying Proposition \ref{ChernoffProp} to $(u \cdot z)$ (with $\E[Y] = N^2p$, and $\alpha =\frac{(\del-1) (u \cdot x)}{2 N}$), we get
\begin{align}
\label{eq:doteven}
\Pr [ (u \cdot z) < N (\del - 1)p (u \cdot x) /2 ] &\le \exp \left [ - \frac{ N^2 p (\del -1)^2 (u \cdot x)^2 }{12 N^2   }   \right ]  \\
\nonumber
&= \exp \left[  - \frac{ K \log N  (u \cdot x)^2}{12N }\right  ] = o\left(\frac{1}{N}\right).
\end{align}

Now we round $z$ to a $\pm 1$ vector $x^\prime$ as above. Let $Z$ be the number of $j$'s so that $x^\prime_j = u_j$.  Then, conditioning on $u\cdot z$ and $\max |z_j|$ as above,
\begin{align*}
\E[Z] &= \sum_{j =1}^N \left(\frac{1}{2} + \frac{ u_j z_j } {2 \max |z_j| }    \right)\\
&= \frac{N}{2} + \frac{ u\cdot z  }{ 2 \max |z_j| } \\
&\ge \frac{N}{2} + \frac{ N (\del - 1)p (u \cdot x)   }{ 4 ( (\del -1) p(u \cdot x) + \frac{(\del - 1) Np}{25} )}.
\end{align*}
If $(\del -1) p (u \cdot x) \le \frac{(\del-1)Np}{25}$, we have
\begin{align*}
\E[Z] & \ge  \frac{N}{2} + \frac{ N (\del - 1)p (u \cdot x)   }{ 8 (\del-1) Np /25 } \ge \frac{N}{2}  + 3 (u \cdot x).
\end{align*}
If $(\del-1) p(u \cdot x) \ge  \frac{(\del-1)Np}{25}$, we have
\begin{align*}
\E[Z] & \ge  \frac{N}{2} + \frac{ N (\del - 1)p (u \cdot x)   }{ 8(\del-1) p (u \cdot x) } = \frac{5 N}{8}.
\end{align*}

Note that the variance of $Z$ is at most $N/4$.  From Chebyshev's inequality, with probability $1- O(N/ (u \cdot x)^2  )$, $Z \ge \min \left \{ \frac{N}{2}  + (u \cdot x)  , \frac{5 N}{9}  \right \}$, which completes the proof of Proposition \ref{evenauxprop}.
\end{proof}

\paragraph{Finishing:}
We consider two phases.  When $|u \cdot x| < \sqrt N \log \log N$,  with probability at least $1/2$, $|u\cdot x^{i+1}| \ge \max \{ \sqrt N/10  , 2|u \cdot x^i| \}$.  This follows from Berry-Esseen bounds in the Central Limit Theorem: $Z$ is the sum of $N$ independent $0,1$ random variables with different probabilities, and we know at least $9N/10$ have a probability between $2/5$ and $3/5$ (comparing a typical $|z_i|$ with $\max |z_i|$).  This shows the variance of $Z$ is at least $N/5$ when $u\cdot x$ is this small.

Now call a step `good' if $|u\cdot x^{i+1}| \ge \max \{ \sqrt N/10  , 2|u \cdot x^i| \}$.  Then in $\log  N$ steps whp there is at least one run of at least $\log \log N$ good steps, and after any such run we have $|u \cdot x| \ge \sqrt N \log \log N$ with certainty, completing the first phase.

 Once we have $|u \cdot x| \ge \sqrt N \log \log N$, then according to Proposition \ref{evenauxprop}, after $O( \log N )$ steps the value of $|x^u \cdot u|$ successively doubles with error probabilities that are geometrically decreasing, and so whp at the end we have a vector $x \in \{ \pm 1\}^{N}$ so that $|u \cdot x | \ge \frac{N}{9}$.  In the positive case, when we multiply $x$ once more by $M \sim M_{\sigma,p}$, we have for $i: u_i =1$, $\E[(Mx)_i] \ge (\del -1) p N /9 $.  Using Proposition \ref{ChernoffProp} (with $\E[Y] =Np$ and $\alpha =(\del-1)/10$),
\begin{align*}
\Pr[ (Mx)_i \le 0 ] & \le e^{-cNp} =o( N^{-2})
\end{align*}
Similarly, if $u_i = -1$, $\Pr[(Mx)_i \ge 0] = o(N^{-2})$, and thus whp rounding to the sign of $x$ will give us $u$ exactly.  The same holds in the negative case where we will get $-u$ exactly.

\subsection{Algorithm Discrete-Power-Iterate (odd $k$)}

\begin{enumerate}
\item Pick $x^0 \in \{ \pm 1 \}^{N_2}$ uniformly at random.  For $i = 1, \dots \log N$, repeat the following:
\begin{enumerate}
\item Draw a sample matrix $M \sim M_{\sigma,p}$.
\item Let $\overline{y}^i = M x^{i-1}$; round $\overline{y}^i$ to a vector $y^i$ with entries $0, +1,$ or $-1$, according to the sign of the coordinates.
\item Draw another sample  $M \sim M_{\sigma,p}$.
\item Let $\overline{x}^i = M^T y^i $. Randomly round each coordinate of $\overline{x}^i$ to $\pm 1$ as follows to get $x^i$:
\[
x^i_j = \begin{cases}
\sgn(\overline{x}_j) \text{ with probability }   \frac{1}{2} + \frac{ |\overline{x}_j|  } {2 \max_j |\overline{x}_j| }   \\
-\sgn(\overline{x}_j) \text{ otherwise.}
\end{cases}
\]
\end{enumerate}
\item Set $u^* = \sgn(x^{\log N})$ by rounding each coordinate to its sign.
\item Output the solution by solving the system of parity equations defined by $u^*$.
\end{enumerate}

\begin{lem}
\label{newoddlem}
Set $p = \frac{K \log N}{(\del-1)^2 N}$.  Then whp, the algorithm returns the planted assignment.
\end{lem}

We will keep track of the inner products $x^i \cdot u_x$ and $y^i \cdot u_y$ as the algorithm iterates.

\begin{prop}
\label{oddxtoyprop}
If $|x^i \cdot u_x| = \beta N_2 \ge \sqrt{N_2}/\log\log N$, then with probability $1- o(1/\log N)$,

\begin{enumerate}
\item $|y^{i+1} \cdot u_y| \ge  N \beta \log N $
\item $ \|y^{i+1}\|_1 = N^2 p(1  +o (1)) $.
\end{enumerate}
\end{prop}

\begin{proof}
Let $x \in \{ \pm 1 \}^{N_2}$ and $M \sim M_{\sigma,p}$.  Let $y = M x$.  We will assume $\del >1$ and $x \cdot u_x >0$ for simplicity.

If $j \in U_y^+$, then
\begin{align*}
\Pr[ y_j \ge1 ] &= \del p |X^+ \cap U_x^+| + (2-\del) p|X^+ \cap U_x^-| + O((N_2-N_2^*)p) + O(p^2 N_2) , \\
\text{and} \\
\Pr[ y_j \le-1 ] &=\del p |X^- \cap U_x^+| + (2-\del) p|X^- \cap U_x^-| +O((N_2-N_2^*)p) + O(p^2 N_2),
\end{align*}
and similarly for $j \in U_y^-$:
\begin{align*}
\Pr[ y_j \ge1 ] &= \del p |X^+ \cap U_x^-| + (2-\del) p|X^+ \cap U_x^+| + O((N_2-N_2^*)p) + O(p^2 N_2),  \\
\text{and} \\
\Pr[ y_j \le-1 ] &=\del p |X^- \cap U_x^-| + (2-\del) p|X^- \cap U_x^+| + O((N_2-N_2^*)p) + O(p^2 N_2).
\end{align*}
Rounding $y$ by the sign of each coordinate gives a $0, +1, -1$ vector $y^\prime$.  Let $Y^+$ be the set of $+1$ coordinates of $y^\prime$, and $Y^-$ the set of $-1$ coordinates.  An application of Proposition \ref{ChernoffProp} with $\E[Y] = N^2p $ and $\alpha =1/ \log N$ immediately gives $\| y^\prime \|_1 =N^2p (1+o(1)) $ with probability $1-o(N^{-1})$.

 We can write
\begin{align*}
y^\prime \cdot u_y &= | Y^+ \cap U_y^+| +  | Y^- \cap U_y^-| -| Y^+ \cap U_y^-| -| Y^- \cap U_y^+|,
\end{align*}
and compute
\begin{align*}
\E[y^\prime \cdot u_y] &= \frac{N_1^*}{2} \left[ (2 \del -2) (|X^+ \cap U_x^+| + |X^- \cap U_x^-|  -|X^+ \cap U_x^-| - |X^- \cap U_x^+|  )  \right ] + O(N_1 N_2 p^2 )   \\
&= N_1 p (\del -1) (x \cdot u_x)(1+o(1)). 
\end{align*}
Another application of Proposition \ref{ChernoffProp} with $\E[Y] =N_1 N_2 p$ and $\alpha =\frac{ (\del-1) (x \cdot u_x) }{ 2 N_2}$ shows that with probability $1- o(N^{-2})$,
\begin{align}
y^\prime \cdot u_y &\ge N_1 p (\del -1) (x \cdot u_x)/2  = \frac{N \beta C \log N }{2 (\del-1)} \ge N \beta \log N . \label{333}
\end{align}
 \end{proof}

 \begin{prop}
\label{oddytoxprop}
If $|y^i \cdot u_y| = \gamma N \ge \sqrt{N_1} \log N/ \log \log N$ with $\|y\|_1 =N^2p (1+o(1)) $, then with probability $1- o(1/\log N)$,
\[ |x^{i} \cdot u_x| \ge \min \left \{ \frac{ N_2}{9}, \frac{ N_2 c \gamma }{ \sqrt{\log N} } \right \}. \]
for some constant $c = c(\del, K)$.
\end{prop}
\begin{proof}
For $j \in U_x^+$ as above we calculate,
\begin{align*}
\E[x_j] &= \del p |U_y^+ \cap Y^+| + (2-\del)p | U_y^- \cap Y^+| - \del p |U_y^+ \cap Y^-|\\
& - (2-\del)p |U_y^- \cap Y^-|  - p \left(|Y^+| - |Y^-|   \right) + O((N_1 - N_1^*)p)\\
&= (\del - 1)p (u_y \cdot y) + O((N_1 - N_1^*)p) \\
&= (\del - 1)p (u_y \cdot y) + O(N_2p)
\end{align*}
 And for $j \in U_x^-$, $\E[x_j] = - (\del-1) p (u_y \cdot y)+ O(N_2p)$.  We also have $\E[u_x \cdot x] = (\del - 1) N_2^* p (u_y \cdot y)  $.

 Proposition \ref{ChernoffProp} with $\E[Y] = N_1p$ and $\alpha = \frac{ N_2 p}{\sqrt {\log N}}$ shows that with probability $1 - o(N^{-1})$,
 \begin{align*}
\max_j |x_j| &\le |\del -1| p (u_y \cdot y) + \frac{N^2p^2}{\sqrt{\log N}},
\end{align*}
and applied with $\E[Y] =N_1 N_2^2 p^2$ and $\alpha = \frac{ 1}{\sqrt{ N_2 \log N}}$ shows that with probability $1 - o(N^{-1})$,
\begin{align*}
(u_x \cdot x) &\ge (\del - 1) N_2 p (u_y \cdot y)  - \frac{N^2 p^2 \sqrt{N_2}}{\sqrt {\log N}} \\
&= (\del - 1) N_2 p (u_y \cdot y) (1+ o(1))
\end{align*}
for $(u_y \cdot y) \ge \sqrt{N_1} \log N/ \log \log N$.
Again we randomly round to a vector $x^*$, and if $Z$ is the number of of coordinates on which $x^*$ and $u_x$ agree,
\begin{align*}
\E[Z] &= \frac{N_2}{2} + \frac{ u_x\cdot x  }{ 2 \max |x_j| } \ge \frac{N_2}{2} + \frac{ N_2 (\del - 1)p (u_y \cdot y)   }{ 4 ( (\del -1) p(u_y \cdot y) +\frac{N^2p^2}{\sqrt{\log N}} )}.
\end{align*}
If $(\del -1) p (u_y \cdot y) \le \frac{N^2p^2}{\sqrt{\log N}}$, we have
\begin{align*}
\E[Z] & \ge  \frac{N_2}{2} + \frac{ N_2 (\del - 1)p (u_y \cdot y)   }{ 8 N^2p^2 /\sqrt{\log N} } =  \frac{N_2}{2} + \frac{ N_2  \gamma (\del-1)^3}{8 K \sqrt{ \log N }}.
\end{align*}
If $(\del-1) p(u \cdot x) \ge \frac{N^2p^2}{\sqrt{\log N}}$, we have
\begin{align*}
\E[Z] & \ge  \frac{N_2}{2} + \frac{ N_2 (\del - 1)p (u \cdot x)   }{ 8(\del-1) p (u \cdot x) } = \frac{5 N_2}{8}.
\end{align*}

Another application of Proposition \ref{ChernoffProp} with $\E[Y] = \E[Z]$ and $\alpha =\frac{ (\del-1)^3 \gamma}{100 K\sqrt{\log N}}$ shows that with probability $1- o(1)$,
\[ Z \ge \min \left \{ \frac{N_2}{2} +  \frac{ N_2  \gamma (\del-1)^3}{9K \sqrt{ \log N }}  , \frac{5 N_2}{9}  \right \} ,  \]
which shows that $x^* \cdot u_x \ge \min \left \{\frac{ N_2 c \gamma}{\sqrt{ \log N }}  , \frac{N_2}{9}\right \}$ for some constant $c= c(\del, K)$.
\end{proof}

\paragraph{Finishing:}

Choosing $x^0$ at random gives $|x^0 \cdot u_x| \ge \frac{\sqrt{N_2}}{\log \log N}$ whp.  After a pair of iterations, Propositions \ref{oddxtoyprop} and \ref{oddytoxprop} guarantee that whp the value of $x^i \cdot u_x$ rises by a factor of $\sqrt{\log N}/K$, so after at most $\log N$ steps we have a vector $x \in \{ \pm 1 \}^{N_2}$ with $| x \cdot u_x | \ge \frac{N_2}{9}$.  One more iteration gives a vector $y$ with $|u_y \cdot y|\ge \frac{N \log N}{9}$.  Now consider $(M^T y)$.  In the positive case (when $u_y \cdot y\ge \frac{N \log N}{9}$),  we have for $i \in U_x^+$, $\E[(M^Ty)_i] \ge (\del -1)  Np \log N /9 $.  Using Proposition \ref{ChernoffProp},   $\Pr[ (M^Ty)_i \le 0 ]   =o( N^{-2})$.  Similarly, if $i \in U^-_x$, $\Pr[(M^T y)_i \ge 0] = o(N^{-2})$, and thus whp rounding to the sign of the vector will give us $u_x$ exactly.  The same holds in the negative case where we will get $-u_x$ exactly.

\subsection{Implementing the algorithms with the statistical oracle}
\label{sec:statImplement}

We complete the proof of Theorem \ref{thm:algo1} by showing how to implement the above algorithms with the statistical oracles $\MSAMPLE$ and $\MVSTAT$.

\begin{lem}[Even $k$]
There is a randomized algorithm that makes $O(N \log^2 N)$ calls to the $\MSAMPLE( N)$ oracle and returns the planted assignment with probability $1-o(1)$.  There is a randomized algorithm that makes $O(\log N)$ calls to the $\MVSTAT(t, L)$ oracle with $L=N$ and $t= N \log \log  N$, and returns the planted assignment with probability $1-o(1)$.
 \end{lem}

\begin{proof}
We can run the above algorithm using the $\MSAMPLE(N)$ oracle.  Given a vector $x \in \{ \pm 1\}^N$, we compute $x^{\prime}$, the next iteration, as follows:  each $j \in [N]$ corresponds to a different value of the query functions $h^+$ and $h^-$ defined as
$h^+(X)=i$ if the clause $X= (i,j)$ for $j: x_j = +1$ and zero otherwise, and similarly $h^-(X) =i$ if $X= (i,j)$ for $j: x_j = -1$ and zero otherwise.  For use in the implementation, we define the Boolean functions $h^+_i$ as $h^+_i(X) =1$ iff $h^+(X)=i$.
Let $v_i^+, v_i^-$ denote the corresponding oracle's responses to the two queries, and $v_i = v_i^+ - v_i^-$.  Now to compute $x^\prime$, for each coordinate  we sum $v_i$ over all samples and subtract $p \sum x_i$.  We use $O(\log N)$ such iterations, and we use $O(N \log N)$ clauses per iteration (corresponding to $p = \frac{K \log N}{(\del-1)^2 N}$).

To use the $\MVSTAT$ oracle, we note that for each query function $v$, we make $t=O(N\log N)$ calls to $\MSAMPLE(N)$. We can replace each group of $t$ calls with a one call to $\MVSTAT(N,t)$. Let the response be a vector $p$ in $[0,1]^L$, with $L+2$ subsets, namely singleton subsets for each coordinate as well as for the subsets with positive parity and with negative parity on the
unknown assignment $\sigma$.  For each coordinate $l$, we set $v_l =  \mbox{Binom}(1,p_l)$, the output of an independent random coin toss with bias $p_l$. The guarantees on $\MVSTAT$ imply that the result of this simulation are equivalent for our purposes to
directly querying $\MSAMPLE$. Here we give a direct simulation with smaller $t$.

For $t= N \log \log  N$, versions of equations (\ref{eq:maxjeven}) and (\ref{eq:doteven}) (properly scaled) hold due to the oracle's bound on $|v_i|$ and the bound on $\sum_V v_i$.

In particular, we can calculate that $\E[h_i^+ - h_i^- - \frac{1}{N} \sum_j x_j] = u_i \cdot \frac{(\del -1)\beta}{N} + O\left (\frac{1}{N (N-N^*)} \right) $ where $u \cdot x = \beta N$.  The oracle bounds then give
$\max_i |v_i| \le \frac{(\del -1)\beta}{N} + \frac{2}{\sqrt{tN}   } =  \frac{(\del -1)\beta}{N} + \frac{2}{N \sqrt{\log \log N}   } $
since $t \gg (N- N^*)$. 
The oracle also guarantees that $\left |u \cdot v - (\del -1)\beta   \right | \le \frac{1}{\sqrt t},$
and so for $\beta \ge \frac{2}{\sqrt{N \log \log N}}$, $u \cdot v \ge (\del-1)\beta /2$.

Now we do the same randomized rounding as above, and we see that

\begin{align*}
\E[Z] &= \sum_{j =1}^N \left(\frac{1}{2} + \frac{ u_j v_j } {2 \max |v_j| }    \right)\\
&= \frac{N}{2} + \frac{ u\cdot v  }{ 2 \max |v_j| } \\
&\ge \frac{N}{2} + \frac{  (\del - 1) \beta   }{ 4 ( \frac{(\del -1)\beta}{N} + \frac{2}{N\sqrt{\log \log N}   }  )}.
\end{align*}

If $\frac{(\del -1)\beta}{N} \le \frac{2}{N \sqrt{\log \log N}   } $, we have
\[ \E[Z]  \ge  \frac{N}{2} + \frac{  (\del - 1) \beta   }{ \frac{16}{N \sqrt{\log \log N}   }  } = \frac{N}{2} + \frac{ \sqrt{\log \log N} (\del - 1)    }{ 16  } \beta N .\]

If $\frac{(\del -1)\beta}{N} \ge \frac{2}{N \sqrt{\log \log N}   } $, we have
\[ \E[Z]  \ge  \frac{N}{2} + \frac{  (\del - 1)\beta   }{ 8(\del-1) \beta /N } = \frac{5 N}{8} .\]

The variance of $Z$ is at most $N/4$, and with probability $1-o(1)$ we start with  $|x^0 \cdot u| \ge \sqrt N/ \log \log \log N$. Then successive applications of Chebyshev's inequality as above show that whp after  at most $\log N$ steps, we have $|x^i \cdot u| \ge \frac{5 N}{8}$.

\end{proof}

\begin{lem}[Odd $k$]
There is a randomized algorithm that makes $O(n^{k/2} \log^2 n)$ calls to the $\MSAMPLE(L)$ oracle for $L = N_1$, and returns the planted assignment with probability $1-o(1)$.  
\end{lem}

\begin{proof}
We run the algorithm using $\MSAMPLE$,  alternately querying $N_1$-valued functions and $N_2$-valued functions, each with $t=O(N\log N)$ samples per iteration. Since there are $O(\log N)$ iterations in all, this gives the claimed bound of $O(N\log^2N)$ calls to $\MSAMPLE(N_1)$.

To implement using $\MVSTAT$, we do as described in proof for the even case. Evaluation of an $L$-valued query $h$ with $t$ samples via $t$ calls to $\MSAMPLE(L)$ is replaced by one call to $\MVSTAT(L,t)$ and this response is used to generate a $0/1$ vector, each time with subsets corresponding to all singletons and the two subsets with different parities according to the planted assignment $\sigma$. This gives the bounds claimed in Theorem~\ref{thm:algo1}. To see that the algorithm converges as claimed, we note that Prop.~\ref{oddxtoyprop} continues to hold, with a lower order correction term in Equation~\eqref{333} for the difference when $y'\cdot u_y$ when $y'$ is obtained by the above simulation. This difference is small as guaranteed by the $\MVSTAT$ oracle on the two subsets corresponding to the positive support and negative support of $u_y$.
\end{proof}

\section{Discussion and open problems}
By querying well-chosen sequences of functions, statistical query algorithms can be efficient and just as powerful as unconstrained algorithmic approaches, in spite of not being able to directly examine samples from an input distribution. As far as we know, there is only one counterexample, namely solving equations over finite fields, which can be done easily by Gaussian elimination but not with any efficient statistical query algorithm. Here we have given a unifying model of planted constraint satisfaction problems and characterized their SQ complexity. Our bounds correspond closely to known upper bounds for unconstrained algorithms.

Our work also gives a new technique for proving lower bounds on SQ algorithm that strengthens and generalizes previous techniques. It has already been crucial in getting tight lower bounds on SQ complexity of stochastic linear optimization and high-dimensional mean estimation \cite{FeldmanGV:15}. It also served as a step toward a characterization of the SQ complexity of solving general problems over distributions given in \cite{Feldman:16sqd}.

We conclude with some candidate directions for future research.
\begin{enumerate}
\item A long-standing and intriguing question is to find an additional example (besides solving equations over finite fields) of a natural problem over distributions for which there exists an efficient algorithm that beats the lower bound for statistical algorithms, and disproves our conjecture.
\item Which additional problems can be addressed using the methods of this paper? One interesting candidate is the problem of detection in a stochastic block model with $k>2$ blocks? There is currently a gap between the information-theoretic and algorithmic thresholds for the number of edges needed for detection, but the gap is only a factor of roughly $k/\log k$. A special case of this problem is planted $k$-coloring.
\item It would be interesting to better understand the relationship of our lower bounds for convex program relaxations to those known for hierarchies of LP and SDP relaxations. Does there exist a unifying approach?
\end{enumerate} 

\section*{Acknowledgments} We thank Amin Coja-Oghlan, Florent Krzakala, Ryan O'Donnell, Prasad Raghavendra, and Lenka Zdeborov\'{a} for insightful comments and helpful discussions. We also thank Jan Vondrak for the proof idea of Lemma \ref{lem:simulate-1-mstat}.

\begingroup
\raggedright
\sloppy
\newcommand{\etalchar}[1]{$^{#1}$}

\endgroup

\end{document}